\documentclass[12pt,reqno]{amsart}
\usepackage[utf8]{inputenc}

\usepackage[english]{babel}
\usepackage[T1]{fontenc}
\usepackage{amssymb}
\usepackage{todonotes}
\usepackage{fullpage}
\usepackage{hyperref}
\usepackage{comment}
\usepackage{tikz}
\usepackage{enumitem}
\usepackage[indent]{parskip}

\theoremstyle{plain}
\newtheorem{lemma}{Lemma}[section]

\newtheorem{proposition}{Proposition}[section]

\newtheorem{theorem}{Theorem}[section]
\newtheorem{corollary}{Corollary}[section]

\theoremstyle{definition}
\newtheorem{definition}[theorem]{Definition}
\newtheorem{example}[theorem]{Example}

\theoremstyle{remark}
\newtheorem{remark}{Remark}[section]

\def\Abox{{\tikz[scale=0.007cm] \draw (0,0) rectangle (1,1);}}
\def\sAbox{{\tikz[scale=0.005cm] \draw (0,0) rectangle (1,1);}}

\numberwithin{equation}{section}

\def\be{\begin{equation}}
\def\ee{\end{equation}}
\def\ba{\begin{aligned}}
\def\ea{\end{aligned}}

\newcommand\mathd{\mathrm{d}}
\newcommand\mathe{\mathrm{e}}
\newcommand\mathi{\mathrm{i}}

\newcommand\BC{\mathbb{C}}
\newcommand\BZ{\mathbb{Z}}

\newcommand\Tr{\mathrm{Tr}}
\newcommand\lam{\lambda}

\newcommand\hG{\hat{G}}
\newcommand\hK{\hat{K}}
\newcommand\hF{\hat{F}}
\newcommand\hD{\hat{D}}
\newcommand\hW{\hat{W}}
\newcommand\hL{\hat{L}}
\newcommand\bp{\mathbf{p}}
\newcommand\bx{\mathbf{x}}
\newcommand\by{\mathbf{y}}
\newcommand\sfP{\mathsf{P}}
\newcommand\sfZ{\mathsf{Z}}
\newcommand\sfW{\mathsf{W}}
\newcommand\sfT{\mathsf{T}}
\newcommand\sfU{\mathsf{U}}
\newcommand\sfA{\mathsf{A}}
\newcommand\sfB{\mathsf{B}}
\newcommand\sfC{\mathsf{C}}

\newcommand\sfc{\mathsf{c}}

\newcommand\ve{\varepsilon}
\newcommand\vac{\varnothing}
\newcommand\res{\mathop{\text{Res}}}
\newcommand\rCS{\mathrm{rCS}}
\newcommand\Nf{N_{\mathrm{f}}}
\newcommand\Nbf{N_{\bar{\mathrm{f}}}}
\newcommand\DIM{U_{q,t}(\hat{\hat{\mathfrak{gl}}}_1)}

\newcommand{\ev}[1]{\Big\langle #1 \Big\rangle}
\newcommand{\dev}[1]{\Big\langle\!\!\Big\langle #1 \Big\rangle\!\!\Big\rangle}

\usepackage{xcolor}
\definecolor{nicklethistwink}{rgb}{0.69, 0.19, 0.38}
\definecolor{ticklemepink}{rgb}{0.99, 0.54, 0.67}
\definecolor{mistyrose}{rgb}{1.0, 0.89, 0.88}
\definecolor{lolololol}{rgb}{0.25, 0.5, 0.75}
\definecolor{upsdellred}{rgb}{0.68, 0.09, 0.13}
\definecolor{ghostwhite}{rgb}{0.97, 0.97, 1.0}
\definecolor{ivory}{rgb}{1.0, 1.0, 0.94}
\definecolor{bubblegum}{rgb}{0.99, 0.76, 0.8}
\definecolor{honeydew}{rgb}{0.94, 1.0, 0.94}
\definecolor{magnolia}{rgb}{0.97, 0.96, 1.0}
\definecolor{isabelline}{rgb}{0.96, 0.94, 0.93}
\definecolor{bubbles}{rgb}{0.91, 1.0, 1.0}
\definecolor{floralwhite}{rgb}{1.0, 0.98, 0.94}
\definecolor{mintcream}{rgb}{0.96, 1.0, 0.98}
\definecolor{champagne}{rgb}{0.97, 0.91, 0.81}
\definecolor{bubblegum}{rgb}{0.99, 0.76, 0.8}

\begin{document}

\title{Superintegrability of $q,t$-matrix models and quantum toroidal algebra recursions}

\author{Luca Cassia}
\address{School of Mathematics and Statistics, The University of Melbourne, Parkville, VIC, 3010, Australia}
\email{luca.cassia@unimelb.edu.au}
\thanks{L.C.\ was supported in part by the ARC Discovery Grant DP210103081, and by the Swedish Research Council under grant no.\ 2021-06594 while in residence at Institut Mittag--Leffler in Djursholm, Sweden during the spring semester of 2025.
}

\author{Victor Mishnyakov}
\address{Nordita, KTH Royal Institute of Technology and Stockholm University, Hannes Alfv\'ens v\"ag 12, SE-106 91 Stockholm, Sweden}
\email{mishnyakovvv@gmail.com}
\thanks{V.M.\ was supported in part by NordForsk}


\date{\today}



\begin{abstract}
$q,t$-deformed matrix models give rise to representations of the deformed Virasoro algebra and more generally of the quantum toroidal $\mathfrak{gl}_1$ algebra. These representations are described in terms of finite difference equations that induce recursion relations for correlation functions. Under suitable assumptions, these recursions admit unique solutions expressible through ``superintegrability'' formulas, i.e.\ explicit closed formulas for averages of Macdonald polynomials.
In this paper, we discuss examples arising from localization of 3d $\mathcal{N}=2$ theories,
which include $q,t$-deformation of well known classical ensembles: Gaussian, Laguerre and Jacobi.
We explain how relations in the quantum toroidal algebra can be used to give a new and universal proof of the known superintegrability formulas, as well as to derive new formulas for models that have not been previously studied in the literature.
Finally, we make some remarks regarding the relation between superintegrability and orthogonal polynomials.
\end{abstract}

\maketitle

\tableofcontents

\section{Introduction}

In this paper we study certain eigenvalue-like integrals that arise as $q,t$-deformations of ordinary Hermitian matrix ensembles, namely
\begin{equation}
 \oint_{\mathcal{C}} \prod_{i=1}^{N} \mathd x_i\,
 \prod_{1\leq i\neq j\leq N}\frac{(x_i/x_j;q)_\infty}{(t x_i/x_j;q)_\infty}
 \, f(x_1,\dots,x_N)
 \prod_{i=1}^N x_i^{\frac{\log t}{\log q}(N-1)} w(x_i)\,.
\end{equation}
These models appear in various contexts, in particular, as a result of supersymmetric localization of certain three-dimensional gauge theories or in $q$-deformed CFTs.
In the context of multivariable orthogonal polynomials, these integrals give rise to inner products on the space of symmetric polynomial functions and, for specific choices of the weight function $w(x)$, they fit into the multivariable generalization of the $q$-Askey scheme.

More concretely, the matrix models considered in this paper are of two kinds, distinguished by their physical origin.
The first type corresponds to a family of integrals that can be given a triality of interpretations, in the sense of \cite{Aganagic:2013tta,Aganagic:2014oia}. These can be labeled by the $3d$ $\mathcal{N}=2$ gauge theory on $D^2 \times_q S^1$ from which they are obtained as a result of supersymmetric localization.
The resulting integral formulas can be understood as the Dotsenko--Fateev representation of  $q$-Virasoro conformal blocks. The third description follow by a direct residue calculation, which shows that evaluation of these integrals produces sums over integers partitions analogous to those that compute the $5d$ Nekrasov partition functions.
In what follows we restrict ourselves to theories with a single $U(N)$ gauge group and not more than two (anti-)fundamental matter multiplets.
Integrals in this class represent deformations and generalizations of well-known classical ensembles, such as the $q$ (or $q,t$)-Gaussian, $q$-Selberg and $q$-Laguerre models.

The second type of integral we consider represents a special standalone case---the matrix integral associated to the refined Chern--Simons (rCS) theory, as defined by Aganagic and Shakirov \cite{Aganagic:2012ne}. This theory is proposed as a deformation or refinement of the standard Chern--Simons theory on $S^3$. At $t=q$ one obtains the so called Stieltjes--Wigert ensemble, which is the matrix model associated to Chern--Simons theory \cite{Tierz:2002jj}.
Expectation values in this model compute the colored HOMFLY--PT quantum knot invariants of the unknot, and in the refined case, the superpolynomials of certain knot homology theories \cite{Dunfield:2005si}.

Computing observables in the physical theory corresponds to evaluating expectation values of symmetric functions in the matrix model. In certain cases, these models exhibit a remarkable property: in a special basis of the space of symmetric functions—the Macdonald polynomials—the expectation values take exceptionally simple forms and can be computed exactly.
Although examples of this property were known previously, the phenomenon appears to be quite universal and has therefore been given a special name: \emph{superintegrability} \cite{Mironov:2022fsr}. For the models considered in this paper, it manifests in the following way.
Expectation values of Macdonald symmetric functions $P_\lam(\bx)$ are given by
\be
\label{eq:SI-intro}
 \ev{P_\lam(\bx)}^w  = \sfC_\lam^w\, P_\lam(p_k = \varphi^{w}_k)
\ee
where $\varphi^{w}_k$ is a certain evaluation locus and the function $\sfC_\lam^w$ is a factorized expression, i.e.\ a product over boxes of the partition $\lam$:
\be
 \sfC_\lam^w = \prod_{(i,j)\in\lam} g^w(q^{j-1}t^{-i+1})\,.
\ee
Both $\varphi^{w}_k$ and the function $g^w$ depend explicitly on the choice of weight function $w(x)$.

In Theorem~\ref{thm:SIqt}, we give a list of ten different $q,t$-models for which superintegrability holds and we list explicitly the corresponding values of the coefficients $\sfC_\lam^w$ and the evaluations $\varphi^{w}_k$.

Different special cases of this formula have been known in the literature. In particular, for the refined Chern--Simons model, this formula is a consequence of the celebrated Cherednik--Macdonald--Mehta identities \cite{Cherednik:1997,Etingof:1997} for root systems of A-type. One case, corresponding to the $q$-Selberg integral was studied and proven by Kaneko \cite{kaneko1996q}, while a formula for the $q,t$-Gaussian model was conjectured by Morozov, Popolitov and Shakirov \cite{Morozov:2018eiq} and recently proven by Forrester and Byun \cite{Byun:2025qrv}. Our first goal in this paper is to provide a comprehensive classification of all such $q,t$-matrix models, that can be shown to have superintegrability using available techniques. The main feature of our approach to prove superintegrability is that it applies universally to all the cases considered.

Our approach is based on solving differential/difference equations known as Ward identities. In the context of classical matrix models, these are usually known as the Virasoro constraints or loop equations. These are equations that can be phrased either as certain linear recursion relations between correlation functions or as certain differential/difference equations satisfied by the generating functions $\sfZ^w(\bp)$. In the deformed context, these are known as $q,t$-Virasoro constraints \cite{Nedelin:2015mio}. Using a specific choice of resummation of such constraint equations, we are able to obtain a single more tractable ``\emph{recursion operator}'' equation of the form
\be
 \hat{\sfA}^w\cdot\sfZ^w(\bp) = 0\,.
\ee
Explicitly solving these equations allows us to derive and prove explicit expressions for the superintegrability data $\sfC_\lam^w$ and $\varphi^{w}_k$ in \eqref{eq:SI-intro}.

Another crucial feature of our approach is that it reveals another hidden algebraic structure behind these matrix models. It turns out, that the $q,t$-Virasoro algebra is actually enhanced to the full quantum toroidal algebra $\DIM$.
It is in fact natural to express the operator $\hat{\sfA}^w$ as a linear combinations of generators of $\DIM$ in the (level-1) Fock representation.
This is of course not the first observation of the role played by this algebra in the context at hand. On the physical side it is expected that quantum toroidal algebras are the hidden symmetries behind the AGT duality. In particular the $q,t$-Virasoro Ward identities are examples of $qq$-character relations \cite{Mironov:2016yue,Awata:2016riz,Nekrasov:2015wsu}.

The refined Chern--Simons model deserves special attention because of its relation to knot invariants. Surprisingly, the $\DIM$-valued recursion operator which annihilates the generating function has a clear geometric meaning in this case. We show, that in the unrefined limit (i.e.\ $t=q$) the recursion equation reduces to the skein recursion relation associated to the unknot, that was recently studied from the perspective of the skein algebra of the torus in \cite{Ekholm:2019yqp,Ekholm:2020csl,Ekholm:2024ceb}. The refined recursions that we obtain here, have not yet been derived geometrically. On the other hand it is conjectured that the refinement (corresponding to categorification in that case) of the skein algebra of the torus is the quantum toroidal $\mathfrak{gl}_1$ algebra. Hence, it is natural to conjecture that the recursion operators obtained in this paper are the correct categorification of the usual skein recursion relations.

Our method for proving superintegrability reformulates the problem entirely in algebraic terms. One of the byproducts of this reformulation is a relation of our approach to the theory of multivariate orthogonal polynomials. Going further along this direction, we establish a connection between the ``\emph{gauge transformation operators}'', $\hG_w$, which act as construction operators for the generating function of superintegrable models and the exponential operators that construct orthogonal symmetric functions from ordinary Macdonald functions $P_\lam$. In particular, we argue that whenever the recursion operator $\hat{\sfA}^w$ can be written as\footnote{By ``Macdonald operator'' we mean any linear combination of horizontal generators of $\DIM$.}
\be
 \hat{\sfA}^w = \hG_w\cdot(\text{Macdonald operator})\cdot\hG_w^{-1}\,,
\ee
then the multivariable orthogonal functions associated to the matrix model integral are given by
\be
 \sfW_\lam^w = (\hG_w^\perp)^{-1}\cdot P_\lam\,,
\ee
and one can read the superintegrability data from the action of $\hG_w$ on the vacuum vector, i.e.\
\be
 \hG_w\cdot1 = \sum_\lam
 \frac{\sfC_\lam^w\,P_\lam(p_k=\varphi^w_k)}{\langle P_\lam,P_\lam\rangle_{q,t}}
 P_\lam(\bp) = \sfZ^w(\bp)\,.
\ee

\paragraph{\textbf{Outline of the paper}}
After collecting some notations and conventions on symmetric functions in Section~\ref{sec:prelim}, we review superintegrability formulas for classical Hermitian matrix models and their operator formulation in Section~\ref{sec:classical}. We use this section to illustrate the general idea on how to use recursion relations to prove superintegrability. Next, we introduce the $q,t$-deformed matrix models in Section~\ref{sec:qt-deformation}, where we briefly discuss peculiarities about contour choices and derive the corresponding deformed Virasoro constraints. Section~\ref{sec:macdonald-SI} contains the main results of this paper. We first state and then prove superintegrability case by case by explicitly solving the recursions.
Finally, in Section~\ref{sec:orthogonal} we conclude with some remarks regarding the relation between superintegrability and orthogonal polynomials for the models considered.

\paragraph{\textbf{Acknowledgments}}
The authors would like to thank Maxim Zabzine for illuminating discussions and collaboration on related projects.

\section{Preliminaries on symmetric functions and Young tableaux}
\label{sec:prelim}

An integer partition $\lam$ is a collection of weakly decreasing non-negative integer numbers $\lam=[\lam_1,\lam_2,\dots]$, $\lam_1\geq\lam_2\geq\dots$, where only a finite number of entries are non-zero. We denote as $\ell(\lam)$ the \emph{length} of $\lam$, i.e.\ the number of non-zero entries, while $|\lam|=\sum_{i\geq1}\lam_i$ denotes the \emph{size} of the partition.
Associated to each partition $\lam$ we can define the standard combinatorial factors $n_\lam$ and $z_\lam$ as
\be
 n_\lam = \sum_{i=1}^{\ell(\lam)}(i-1)\lam_i\,,
 \hspace{30pt}
 z_\lam = \prod_{i\geq1} m_i! \cdot i^{m_i}
\ee
where $m_i$ is the multiplicity of the part $i$ in $\lam$.

Integer partitions are in one-to-one correspondence with Young tableaux for which we adopt the English convention. We denote as $\lam'$ the conjugate partition to $\lam$, which is the partition associated to the Young tableau obtained via a reflection along the diagonal.
Moreover, for each box $\Abox$ in the tableau, we introduce the symbol $\chi_\sAbox$ to denote the \emph{content} of the box, i.e.\
\be\label{eq:box-content}
 \chi_\sAbox = q^{j-1}t^{-i+1}
\ee
where $q,t\in\BC^\times$ are two complex parameters and $(i,j)$ represent the coordinates of the box in the Young tableau: $i$ labels the row and $j$ labels the column.
With this notation we can further define the quantities
\be
\label{eq:symbols1}
 \chi_\lam = \sum_{\sAbox\in\lam} \chi_\sAbox\,,
 \hspace{30pt}
 \sfT_\lam = \prod_{\sAbox\in\lam} \chi_\sAbox
 = q^{n(\lam')}t^{-n(\lam)}\,,
\ee
\be
\label{eq:symbols2}
 \ve_\lam = \sum_{i\geq1} q^{\lam_i} t^{-i}\,,
 \hspace{30pt}
 x_\lam = 1-(1-q)(1-t^{-1})\chi_\lam = \frac{\ve_\lam}{\ve_\vac}\,,
\ee
and
\be
\label{eq:symbols3}
 u_\lam = \sum_{i=1}^N q^{\lam_i} t^{N-i} = \frac{1-t^N x_\lam}{1-t}\,,
 \hspace{30pt}
 N\geq\ell(\lam)\,.
\ee
Each of these can be regarded as a Laurent polynomial/series in the variables $q,t$.
Then, for a generic symbol $s=s(q,t)$ of this type, we introduce the notation for the dual symbol $s^\vee$ as
\be
 s^\vee = s(q,t)^\vee = s(q^{-1},t^{-1})\,.
\ee

Another combinatorial and algebraic object that we will need to introduce is that of \emph{symmetric polynomial} or \emph{symmetric function}. Let $\{x_i\}_{i=1}^N$, be a finite set of formal variables. A symmetric polynomial $f(\bx)$ is an element of the ring $\Lambda_N = \BC[x_1,\dots,x_N]^{S_N}$, where $S_N$ represents the symmetric group which acts by permutations of the $n$ variables. The space $\Lambda_n$ admits several different basis, among which we recall the \emph{power sum} basis labeled by integer partitions as
\be
 p_\lam = \prod_{k\in\lam} p_k\,,
 \hspace{30pt}
 p_k = \sum_{i=1}^N x_i^k\,,
\ee
as well as the \emph{Schur} basis
\be\label{eq:h_kdef}
 s_\lam = \det_{1\leq i,j\leq N} h_{\lam_i-j+N}\,,
 \hspace{30pt}
 \exp\left(\sum_{k\geq1}\frac{z^kp_k}{k}\right) = \sum_{d=0}^\infty z^d h_d
\ee
which we express in terms of the previously defined generators $p_k$'s via the \emph{homogeneous symmetric} polynomials $h_d$. This relation, known as the Jacobi--Trudi formula, is \emph{stable} in the sense that it holds for any $N$. In the limit $N\to\infty$, the generators $p_k$ become linearly independent and they freely generate the ring of \emph{symmetric functions} in infinitely many variables.
The ring $\Lambda = \varprojlim_N \Lambda_N$ can then be naturally identified with the polynomial ring in the power sum functions, i.e.\ $\Lambda\cong \BC[p_1,p_2,\dots]$.

We recall also the celebrated Cauchy identity
\be
\label{eq:Cauchy}
 \exp\left( \sum_{k\geq1} \frac{p_k\otimes p_k}{k} \right)
 = \sum_\lam s_\lam\otimes s_\lam
 = \sum_\lam z_\lam^{-1} p_\lam\otimes p_\lam\,,
\ee
where the l.h.s.\ can be regarded as the reproducing kernel of the Hall inner product defined as
\be
 \langle s_\lam,s_\mu\rangle = \delta_{\lam,\mu}\,,
\ee
which is the $N\to\infty$ limit of the Weyl orthogonality relation for characters of $\mathfrak{gl}_N$,
\be
 \frac{1}{N!}\oint\prod_{i=1}^N\frac{\mathd x_i}{2\pi\mathi x_i}
 \prod_{i\neq j}(1-x_ix_j^{-1})
 s_\lam(\bx) s_\mu(\bx^{-1}) = \delta_{\lam,\mu}
\ee
with $\bx=(x_1,\dots,x_N)$.

A useful fact that we will use in the rest of the paper, is that the adjoint of $p_k$ w.r.t.\ the Hall inner product corresponds to the differential operator $k\frac{\partial}{\partial p_k}$, hence we can write
\be
 \langle p_k f, g\rangle = \langle f, k\frac{\partial}{\partial p_k} g\rangle
\ee
for any $f,g\in\Lambda$.

Another important basis of $\Lambda_N$ that we shall consider is that of Macdonald symmetric polynomials $P_\lam(x;q,t)$. These functions were introduced by I.\ G.\ Macdonald and they can be uniquely defined as the eigenfunctions of the Macdonald operators (trigonometric Ruijsenaars--Schneider Hamiltonians)
\be
 u^\pm_0 = \sum_{i=1}^N \prod_{j\neq i}\frac{t^{\pm1}x_i-x_j}{x_i-x_j}q^{\pm x_i\frac{\partial}{\partial x_i}}\,,
 \hspace{30pt}
 \begin{aligned}
 u^+_0 P_\lam(\bx;q,t) &= u_\lam P_\lam(\bx;q,t)\,,\\
 u^-_0 P_\lam(\bx;q,t) &= u_\lam^\vee P_\lam(\bx;q,t)\,,
 \end{aligned}
\ee
normalized such that $P_\lam(\bx;q,t)=m_\lam(\bx)+\dots$, where ``\dots'' represent lower order terms w.r.t.\ some triangular structure and $m_\lam$ are the monomial symmetric functions.
In the stable limit $N\to\infty$, these symmetric polynomials give rise to analogous symmetric functions $P_\lam\in\Lambda$. The Macdonald difference operators $u^\pm_0$ are replaced by the zero-modes of the vertex operators $x^\pm(z)=\sum\limits_{k\in\BZ}z^k x^\pm_{-k}$, with
\be\label{eq:VertexOperatorsx}
 x^\pm(z) = \exp\left(\sum_{k\geq1}(1-t^{\mp k})z^k\frac{p_k}{k}\right)
 \exp\left(-\sum_{k\geq1}(1-q^{\pm k})z^{-k}\frac{\partial}{\partial p_k}\right)
\ee
and we have
\be
 x^+_0\cdot P_\lam = x_\lam P_\lam\,,
 \hspace{30pt}
 x^-_0\cdot P_\lam = x_\lam^\vee P_\lam
\ee
with eigenvalue $x_\lam$ as in \eqref{eq:symbols2}. The normalization condition can be expressed as $\langle s_\lam,P_\lam\rangle=1$.
In addition to Macdonald operators, the symmetric functions $P_\lam$ provide a basis of eigenfunctions for the \emph{nabla} operator or \emph{framing} operator introduced by \cite{Bergeron:1998sci},
\be
\label{eq:framing}
 \sfT\cdot P_\lam = \sfT_\lam P_\lam\,.
\ee
An important property of the framing operator is that it satisfies
\be
\label{eq:BHid}
 \sfT \exp\left( \sum_{k\geq1} \frac{p_k}{k(1-q^k)} \right)
 = \exp\left( - \sum_{k\geq1} \frac{(-1)^k p_k}{k(1-q^k)} \right)\,.
\ee
Similarly, we define the \emph{delta} operators (as introduced in \cite{Garsia:2018fiv}) as the diagonal operators
\be\label{eq:DeltaOp}
 \Delta^\pm(z)\cdot P_\lam = \Delta^\pm_\lam(z) P_\lam\,,
 \hspace{30pt}
 \Delta^\pm_\lam(z) = \prod_{\sAbox\in\lam} (1-z\chi_\sAbox^{\pm1})\,,
\ee
where we also observe that the eigenvalue $\Delta^{+}_\lam(z)$ can be written as
\be
 \Delta^{+}_\lam(z) = \frac{P_\lam\left(p_k=\frac{1-z^k}{1-t^k}\right)}
 {P_\lam\left(p_k=\frac{1}{1-t^k}\right)}\,,
\ee
which follows from \cite[Ch.VI, \textsection6, (6.17)]{Macdonald:book}. It is also useful to recall that
\be\label{eq:Delta_id}
 \Delta^{+}(z)=(-z)^{\hD}\sfT\Delta^{-}(z^{-1})\,,
\ee
where $\hD:=\sum\limits_{k\geq1} k p_k \frac{\partial}{\partial p_k}$ is the degree operator
\be
\label{eq:degree-op}
 \hD\cdot s_\lam = |\lam|\,s_\lam\,.
\ee

Equivalently, one can define Macdonald functions as the orthogonal basis w.r.t.\ the Macdonald inner product defined as
\be
 \langle p_\lam, p_\mu\rangle_{q,t} = z_\lam \delta_{\lam,\mu} \prod_{k\in\lam} \frac{1-q^k}{1-t^k}
\ee
with normalization fixed by
\be\label{eq:bnormqt}
 \langle P_\lam,P_\lam\rangle_{q,t} = b_\lam^{-1}\,,
 \hspace{30pt}
 b_\lam = \prod_{(i,j)\in\lam}\frac{1-q^{\lam_i-j}t^{\lam_j'-i+1}}
 {1-q^{\lam_i-j+1}t^{\lam_j'-i}}\,.
\ee
Similarly to the case of the Hall pairing, we can define adjoint operators of the power sums $p_k$ as the differential operators $p_k^\perp$,
\be
 \langle p_k f,g\rangle_{q,t} = \langle f,p_k^\perp g\rangle_{q,t}\,,
 \hspace{30pt}
 p_k^\perp = k\frac{1-q^k}{1-t^k}\frac{\partial}{\partial p_k}\,.
\ee
The reproducing kernel associated to the product $\langle \,,\,\rangle_{q,t}$ can be written as
\be
  \exp\left( \sum_{k\geq1}\frac{1-t^k}{1-q^k}\frac{p_k\otimes p_k}{k} \right)
 = \sum_\lam b_\lam P_\lam\otimes P_\lam\,,
\ee
where the equality between the two expressions generalizes the Cauchy identity \eqref{eq:Cauchy} and reduces to it in the limit $t=q$.

The Macdonald inner product can be obtained as the $N\to\infty$ limit of the finite integral
\be\label{eq:macdonald-kernel}
 \langle f,g\rangle'_{q,t}
 = \frac{1}{N!} \oint \prod_{i=1}^N \frac{\mathd x_i}{2\pi\mathi x_i}
 \Delta_{q,t}(\bx) f(\bx) g(\bx^{-1})\,,
 \hspace{30pt}
 \Delta_{q,t}(\bx) = \prod_{i\neq j}\frac{(x_ix_j^{-1};q)_\infty}{(tx_ix_j^{-1};q)_\infty}
\ee
where $\Delta_{q,t}(\bx)$ is the Macdonald deformation of the Vandermonde determinant, and
\be\label{eq:q-pochhammer-def}
 (z;q)_\infty=\prod_{k=0}^\infty(1-q^kz)
\ee
is the $q$-Pochhammer symbol. We then have
\be
 \frac{\langle P_\lam,P_\mu\rangle'_{q,t}}{\langle 1,1\rangle'_{q,t}}
 = (b^{(N)}_\lam)^{-1}\delta_{\lam,\mu}
\ee
with $\lim\limits_{N\to\infty}b^{(N)}_\lam=b_\lam$.

Finally, we introduce some standard notation regarding $q$-analogues of ordinary differential and integral operators.
The $q$-analogue of the derivative w.r.t.\ $x$ is the $q$-derivative $D^q_x$, which is defined as
\be
 D_x^q f(x)
 = \frac1x\frac{1-q^{x\frac{\partial}{\partial x}}}{1-q} f(x)
 = \frac{f(x)-f(qx)}{x(1-q)} \, .
\ee
and the Jackson integral is defined as its formal inverse, i.e.\ the discrete sum
\be
 \int\limits_{0}^{a} \mathd_q x\,f(x)
 = \frac{1-q}{1-q^{a\frac{\partial}{\partial a}}} a f(a)
 = (1-q)\sum_{n=0}^\infty q^{na\frac{\partial}{\partial a}} a f(a)
 = (1-q)\sum_{n=0}^{\infty} (q^n a) f(q^n a) 
\ee
and
\be
 \int\limits_{a}^{b} \mathd_q x\,f(x)
 = \int\limits_{0}^{b} \mathd_q x\,f(x)-\int\limits_{0}^{a} \mathd_q x\,f(x)\,.
\ee
In the limit $q\to1$, one recovers the usual definitions of derivative and integral.

\section{Classical Matrix Models}
\label{sec:classical}

In this section, we review the classical Hermitian random matrix ensembles focusing on the Virasoro constraints that they satisfy, as a preamble to the next sections where we will consider the Macdonald deformation of these models. We use this section to also clarify what we mean by \emph{superintegrability} and we give a definition which will allow us to easily generalize to the Macdonald (refined) case later on.

\subsection{Hermitian matrix models}
We start by considering the random matrix model whose partition function is given by the integral over Hermitian $N$-by-$N$ matrices $M\in H_N$,
\be
 \sfZ^V = \int_{H_N} \mathd M \, \mathe^{-\Tr V(M)}\,,
\ee
with $V(M)$ some function known as the \emph{potential} of the matrix model. We will take $V(M)$ to be invariant under the action of the group $U(N)$ on $H_N$ by conjugation.

This allows us to rewrite the matrix integral as an integral over the eigenvalues $\{x_i\}_{i=1}^N$ of the matrix $M$, and we have
\be
 \sfZ^V = \frac{1}{N!} \int\prod_{i=1}^N \mathd x_i \prod_{i\neq j}(x_i-x_j) \,
 \mathe^{-\sum_{i=1}^N V(x_i)}\,,
\ee
where $\prod_{i\neq j}(x_i-x_j)$ is the Vandermonde determinant.
The observables of the matrix model are the expectation values of operators corresponding to functions of $M$ which are also invariant under the $U(N)$-action. These correspond to insertions of symmetric polynomials in the eigenvalues $x_i$. Hence, we define the normalized correlation functions as
\be
 \ev{ f }^V = \dfrac{ \displaystyle\int_{H_N} \mathd M \, f(M)\, \exp\left({ -\Tr \,V(M)}  \right) }
 { \displaystyle\int_{H_N} \mathd M \, \exp\left({ -\Tr \,V(M)}  \right)}
 = \dfrac{ \displaystyle\int\prod\limits_{i=1}^N\mathd x_i \prod_{i\neq j}(x_i-x_j) \, f(\bx)\,
 \exp \left({-\sum\limits_{i=1}^N V(x_i)} \right) }
 { \displaystyle\int\prod\limits_{i=1}^N\mathd x_i \prod_{i\neq j}(x_i-x_j) \,
 \exp \left({-\sum\limits_{i=1}^N V(x_i)} \right) }\,,
\ee
and it is understood that the matrix model average $\ev{ f }^V$ depends on the choice of potential.

Instead of considering individual correlation functions, it is often useful to collect all of them into a generating function, 
\be
\label{eq:GF}
 \sfZ^V(\bp)
 = \ev{ \exp\Big(\sum_{k\geq 1} \frac{p_k}{k} \Tr(M^k)\Big) }^V
 = \ev{ \exp\Big(\sum_{k\geq 1} \frac{p_k}{k} \sum_{i=1}^N x_i^k \Big) }^V
\ee
where $p_k\in\Lambda$ are regarded as formal generating function parameters also known as \emph{higher times} and $\bp=(p_1,p_2,\dots)$. Because of this, we will sometimes refer to $\sfZ^V(\bp)$ as the \emph{time-dependent} partition function.
 
As a formal power series in $p_k$'s, $\sfZ^V(\bp)$ is not necessarily convergent, however we will assume that $V$ and the contours of integration are chosen is such a way that the coefficients in the $p_k$ expansion are all well-defined convergent integrals. Moreover, because of the normalization of the correlators, we automatically have that $\sfZ^V(0)=1$.

Making use of the Cauchy identity \eqref{eq:Cauchy}, we can rewrite the exponential in the generating function so that we have
\be
\label{eq:gen-f-coeff-c}
 \sfZ^V(\bp)
 = \sum_\lam \sfc_\lam\, \frac{p_\lam}{z_\lam}
\ee
where the sum is over all partitions $\lam$, and the coefficient $\sfc_\lam$ is given by
\be
 \sfc_\lam
 = \ev{ \prod_{k\in\lam} \Tr(M^k) }^V
 = \ev{ \prod_{k\in\lam} \Big( \sum_{i=1}^N x^k_i \Big) }^V\,,
\ee
which are the standard multi-trace correlation functions often of interest in the theory of matrix models. It turns out that these are not easy to evaluate in closed form for generic $\lam$ and $N$. It is indeed often the case that different basis of observables exhibit different properties, and only some of them admit ``nice'' formulas. As we will show later, this choice of basis relates to the definition of superintegrability.

It is useful to introduce another equivalent way to rewrite the expansion, which is through the use of Schur functions, as
\be
\label{eq:GF-Schur}
 \sfZ^V(\bp) = \sum_\lam \ev{ s_\lam(\bx) }^V\, s_\lam(\bp)
\ee
where we regard $s_\lam(\bx)=s_\lam(x_1,\dots,x_N)\in\Lambda_N$ as a symmetric polynomial in the integration variables $x_i$, while $s_\lam(\bp)\in\Lambda$ is regarded as a symmetric function expressed as a polynomial in the time variables $p_k$.

\subsection{Superintegrability}

Matrix model partition functions of the type \eqref{eq:GF} are known to provide examples of $\tau$-functions of integrable hierarchies (see \cite{Alexandrov:2012tr} for a review). This is in fact true of all $\sfZ^V(\bp)$ independently of the choice of potential.
However, for special choices of potentials, additional properties do appear.
As suggested by the Cauchy expansion \eqref{eq:GF-Schur} of the time-dependent partition function, expectation values of Schur polynomials can be used instead of multi-trace correlators. It turns out that these do admit a closed expression as a function of $\lam$, for special choices of $V(M)$.
This property was first observed in \cite{Itoyama:2017xid,Mironov:2017och} and it is now know under the name of \emph{superintegrability}. We will now show this in some examples.

\begin{example}
Let us consider the case of $V(x)=\frac{x^2}{2}$, i.e.\ the \emph{Gaussian} matrix model or \emph{Gaussian Unitary Ensemble} (GUE). We denote matrix expectation values as $\ev{f}^{\rm GUE}$, i.e.\
\be
 \ev{ f }^{\rm GUE}
 = \frac{ \dfrac{1}{N!} \displaystyle\int\prod\limits_{i=1}^N \mathd x_i \prod_{i\neq j}(x_i-x_j)\, f(\bx)\,
 \exp\left(-\sum_{i=1}^N \frac{x_i^2}{2} \right) }
 { \dfrac{1}{N!} \displaystyle\int\prod\limits_{i=1}^N \mathd x_i \prod_{i\neq j}(x_i-x_j)\,
  \exp\left(-\sum_{i=1}^N \frac{x_i^2}{2} \right) }\,.
\ee
The expectation value of Schur polynomials then assumes the following form \cite{Cassia:2020uxy}:
\be
 \ev{ s_\lam }^{\rm GUE}
 = \sfC^{\rm GUE}_\lam \cdot s_\lam(p_k=\delta_{k,2})\,,
 \ee
with 
\be
\label{eq:ClamGUE}
 \sfC^{\rm GUE}_\lam := \prod_{(i,j)\in\lam} (N+j-i)
 = \frac{s_\lam(p_k=N)}{s_\lam(p_k=\delta_{k,1})}\,.
\ee
\end{example}

\begin{example}
Another interesting example is the Wishart--Laguerre (WL) model, $V(x):=x-\alpha\log(x)$, where the expectation values of Schur polynomials have the form \cite{Cassia:2020uxy}:
\be
 \ev{ s_\lam }^{\rm WL}
 = \sfC^{\rm WL}_\lam \cdot s_\lam(p_k=\delta_{k,1})\,,
\ee
with
\be
\label{eq:ClamWL}
 \sfC^{\rm WL}_\lam := \prod_{(i,j)\in\lam} (N+j-i) \prod_{(i,j)\in\lam} (N+\alpha+j-i)
 = \frac{s_\lam(p_k=N)}{s_\lam(p_k=\delta_{k,1})}
 \frac{s_\lam(p_k=N+\alpha)}{s_\lam(p_k=\delta_{k,1})}\,.
\ee
\end{example}

Let us now give a working definition for superintegrability, that we will use in this paper.

\begin{definition}[Schur superintegrability of Hermitian eigenvalue models]
\label{def:SI}
A matrix model is said to be Schur superintegrable if
there exist an operator $\hat\sfC^V:\Lambda\to\Lambda$ diagonal on Schur functions,
\be
 \hat\sfC^V\cdot s_\lam = \sfC_\lam^V \, s_\lam\,,
 \hspace{30pt}
 \text{with eigenvalues}
 \hspace{30pt}
 \sfC_\lam^V = \prod_{(i,j) \in \lam} g^V(i,j)\,,
\ee
for some function $g^V(i,j)$ of the coordinates of the boxes in the partition\footnote{In all known examples, this is a function of $j-i$ only.}, which depends on the choice of potential $V$, and a ring homomorphism $\varphi^{V}: \Lambda \rightarrow \BC$, that acts as a specialization of the power sums to some $V$-dependent locus
such that
\be
 \ev{ f }^V = \varphi^V \left(\hat\sfC^V \cdot f(\bp)\right)\,,
\ee
for any symmetric function $f\in\Lambda$.
\end{definition}
With a slight abuse of notation, we will sometimes write
\be
 \varphi^V_k = \varphi^V(p_k)
\ee
for the evaluation of the $k$-th powersum.
In particular, this means that the averages of Schur polynomials in a superintegrable model can be written as
\be
 \ev{ s_\lam }^V
 = \sfC_\lam^V \, \varphi^V(s_\lam)
 = \sfC_\lam^V \, s_\lam(p_k=\varphi_k^V)\,.
\ee
In the examples above, we have:
\be
 \varphi^{\rm GUE}(p_k) = \delta_{k,2}\,,
 \hspace{30pt}
 g^{\rm GUE}(i,j) = (N+j-i)\,,
\ee
and
\be
 \varphi^{\rm WL}(p_k)= \delta_{k,1}\,,
 \hspace{30pt}
 g^{\rm GUE}(i,j) = (N+j-i)(N+\alpha+j-i)\,.
\ee

\subsection{Virasoro constraints}
\label{sec:Virasoro-constraints-classical}
The operation of computing the matrix model average of a symmetric polynomial in $\Lambda_N$ is linear, hence it can be regarded as a module homomorphism (but not ring homomorphism in general),
\be
 \ev{\,}^V:\Lambda_N \to \BC\,.
\ee
This morphism has a non-trivial kernel corresponding to certain linear relations that the correlation functions satisfy. If we regard the matrix integral as a zero-dimensional path integral, then these relations can be regarded as the Ward identities of the theory.

Equivalently, we can recast such relations as the vanishing of certain differential operators in times $p_k$ acting on the generating function $\sfZ^V(\bp)$. In fact, the differential operators in times act on the generating function by introducing an insertion of the corresponding monomial in power sums in the integration variables, for example
\be
\label{eq:time-derivatives-on-GF}
 n\frac{\partial}{\partial p_n} \sfZ^V(\bp)
 = n\frac{\partial}{\partial p_n}
 \ev{ \exp\Big(\sum_{k\geq 1}\frac{p_k}{k}\sum_{i=1}^N x_i^k\Big) }^V
 = \ev{ \sum_{i=1}^N x_i^n\, \exp\Big(\sum_{k\geq 1}\frac{p_k}{k}\sum_{i=1}^N x_i^k\Big) }^V\,.
\ee
Among all these relations, there is a special infinite family known as the Virasoro constraints which can be obtained through the insertion of a total differential operator in the matrix integral. Namely,
\be
\label{eq:Virasoro1}
 \int\prod_{i=1}^N\mathd x_i \sum_{l=1}^N\frac{\partial}{\partial x_l}
 \left(x_l^{m+1}\prod_{i\neq j}(x_i-x_j)\,\exp\left(-\sum_{i=1}^N V(x_i)+\sum_{k\geq1}\frac{p_k}{k}\sum_{i=1}^N x_i^k \right) \right)=0\,,
\ee
where equality with zero follows from the application of Stokes' theorem.\footnote{We assume that the term inside of the parenthesis vanishes at the boundary of the integration domain whenever such boundary is not empty. If this is not the case, then one obtains a non-homogeneous set of equations \cite{Cassia:2021dpd}.}
Rewriting the derivative in the integrand as the insertion of a linear combination of polynomial functions in $x_i$, and then rewriting each term as a combination of derivatives in times as in \eqref{eq:time-derivatives-on-GF}, we can express the constraints as the differential equations
\be
 \hL_m \cdot \sfZ^V(\bp) =0\,,
\ee
w.r.t.\ the formal variables $p_k$.

The name ``Virasoro constraints'' comes from the fact that the operators $\hL_m$ form a D-module representation of (a parabolic subalgebra of) the Virasoro algebra,
\be
 [\hL_m,\hL_n] =(n-m) \hL_m 
\ee
Because derivatives in times can only reproduce insertions of polynomials in $x_i$ with positive powers, it might happen that some of the constraints obtained as in \eqref{eq:Virasoro1} cannot be rewritten as differential operators acting on $\sfZ^V(\bp)$.
In the example of the GUE, for instance, only the constraints with $m\geq -1$ can be written in terms of operators $\hL_m$, and similarly, in the example of the WL model, we only have $\hL_m$ with $m\geq0$.

The Virasoro constraints can be used to determine the coefficients $\sfc_\lam$ of the series expansion \eqref{eq:gen-f-coeff-c} of the generating function of the matrix models, up to a set of initial conditions. These initial conditions parametrize the possible choices of contour, or equivalently the saddles of the potential $V$.
Here we only focus on the cases when the solution to the constraints is unique up to an overall constant. Since the constraints fix the generating function uniquely, this should imply that superintegrability must follows from the Virasoro constraint equations.
We now show that this is indeed true, at least in the examples of the GUE and WL models.

First, we construct the solution to the constraint equations. One way to find the solution is to rewrite the infinite set of constraints as a single more manageable, but equivalent, equation. This single equation is also known as the $W$-operator constraint \cite{Morozov:2009xk}. Let us briefly recall the idea, since it is going to be instrumental later on for the Macdonald deformation.

\begin{example}
For the GUE, the constraints are well-defined for $m\geq-1$ and the solution is unique up to a constant. The solution can be constructed by re-summing the Virasoro operators
\be
 \hL_m = -(m+2)\partial_{p_{m+2}}
 + \sum_{a+b=m} ab\partial_{p_a}\partial_{p_b}
 + \sum_{k>0} (k+m) p_k \partial_{p_{k+m}}
 + 2 N m\partial_{p_m}
 + N^2 \delta_{m,0}
 + N p_1 \delta_{m,-1}
\ee
as suggested in \cite{Morozov:2009xk}, namely
\be
 -\sum_{m=-1}^\infty p_{m+2} \hL_m
 = \hD - \hW_{-2}\,,
\ee
where $\hD$ is the degree-zero operator \eqref{eq:degree-op} and $\hW_{-2}$ is a degree-two (shifted) cut-and-join operator defined as
\be
 \hW_{-2} = \sum_{a,b=1}^\infty ab p_{a+b+2} \partial_{p_a}\partial_{p_b}
 + \sum_{a,b = 1}^\infty (a+b-2) p_a p_b \partial_{p_{a+b-2}}
 + 2N \sum_{k=1}^\infty k p_{k+2} \partial_{p_k}
 + N^2 p_2 + N p_1^2\,.
\ee
Then Virasoro constraints are equivalent to
\be
\label{eq:virGUEW}
 \left(\hD - \hW_{-2} \right) \cdot \sfZ^{\rm GUE}(\bp) = 0\,.
\ee
We now observe the following: by the Baker--Campbell--Hausdorff formula $\mathe^X Y\mathe^{-X}=Y+[X,Y]+\frac1{2!}[X,[X,Y]]+\dots$, together with $[\hD,\hW_{-2}]=2\hW_{-2}$, we find that we can write
\be
 \hD - \hW_{-2}
 = \hD + [\hW_{-2}/2,\hD]
 = \exp(\hW_{-2}/2) \hD \exp(-\hW_{-2}/2)\,.
\ee
It becomes now evident that \eqref{eq:virGUEW} is a sort of ``\emph{gauge transformed}'' version of the simple equation $\hD\cdot 1=0$, where the gauge transformation is realized by the operator $\exp(\hW_{-2}/2)$.
This implies that
\be
 \sfZ^{\rm GUE}(\bp) = \exp(\hW_{-2}/2)\cdot 1
\ee
is a solution, where $1$ is regarded as the vacuum vector in the Fock module $\Lambda$. A rescaling of the vacuum by an arbitrary constant still gives a solution to the $W$-operator constraint, however we will always fix this constant to $1$ in order to have the normalization property $\sfZ^V(0)=1$.

One may wonder how this is related to superintegrability. In order to answer this question, we observe the following crucial identity \cite{Alexandrov:2014zfa}
\be
 \hat\sfC^{\rm GUE}\cdot p_2\cdot (\hat\sfC^{\rm GUE})^{-1} = \hW_{-2}
\ee
with $\hat\sfC^{\rm GUE}$ the operator diagonal on Schur functions with eigenvalues as in \eqref{eq:ClamGUE}. It then immediately follows that
\be
\ba
 \sfZ^{\rm GUE}(\bp) &= \exp(\hW_{-2}/2)\cdot 1 \\
 &= \hat\sfC^{\rm GUE}\cdot \exp(p_2/2)\cdot (\hat\sfC^{\rm GUE})^{-1} \cdot 1 \\
 &= \hat\sfC^{\rm GUE}\cdot \exp\Big(\sum_{k\geq1}\frac{\delta_{k,2}p_k}{k}\Big) \\
 &= \hat\sfC^{\rm GUE}\cdot \sum_\lam s_\lam(p_k=\delta_{k,2}) s_\lam(\bp) \\
 &= \sum_\lam \hat\sfC^{\rm GUE}_\lam \varphi^{\rm GUE}(s_\lam) s_\lam(\bp) \\
 &= \sum_\lam \varphi^{\rm GUE}(\hat\sfC^{\rm GUE}\cdot s_\lam) s_\lam(\bp)
\ea
\ee
from which we can read the matrix model averages of the Schur observables as
\be
 \ev{ s_\lam }^{\rm GUE} = \varphi^{\rm GUE}(\hat\sfC^{\rm GUE}\cdot s_\lam)
\ee
This is precisely the definition of Schur superintegrability in Definition~\ref{def:SI}.
In this way, we have re-proven that the GUE is superintegrable by using the Virasoro constraints. Notice that this is completely equivalent to the proof of \cite{Wang:2022fxr,Bawane:2022cpd}.
\end{example}

\begin{example}
The same can be done in the WL case. The Virasoro constraints are well-defined for $m\geq0$ \cite[(2.28)]{Cassia:2020uxy} and, after resummation, they give rise to the equivalent equation
\be
 \left(\hD - \hW_{-1}\right)\cdot \sfZ^{\rm WL}(\bp) = 0\,,
\ee
with
\be
 \hW_{-1} = \sum_{a,b=1}^\infty ab p_{a+b+1} \partial_{p_a}\partial_{p_b}
 + \sum_{a,b=1}^\infty (a+b-1) p_a p_b \partial_{p_{a+b-1}}
 + (2N+\alpha)\sum_{k=1}^\infty k p_{k+1}\partial_{p_k}
 + N (N+\alpha) p_1 \,.
\ee
As before, we can rewrite the recursion operator as the conjugate of $\hD$ by some invertible gauge transformation,
\be
 \hD - \hW_{-1}
 = \hD + [\hW_{-1},\hD]
 = \exp(\hW_{-1}) \hD \exp(-\hW_{-1})\,,
\ee
so that we can write the solution as
\be
\label{eq:gauge-transform-WL}
 \sfZ^{\rm WL}(\bp) = \exp(\hW_{-1})\cdot 1\,.
\ee
Moreover, observing that $\hat\sfC^{\rm WL}\cdot p_1 \cdot(\hat\sfC^{\rm WL})^{-1}=\hW_{-1}$, with $\hat\sfC^{\rm WL}$ diagonal on Schur with eigenvalue as in \eqref{eq:ClamWL}, we obtain
$\sfZ^{\rm WL}(\bp)=\hat\sfC^{\rm WL}\cdot\mathe^{p_1}$ which leads to the proof of Schur superintegrability for the WL model, i.e.
\be
 \ev{s_\lam}^{\rm WL} = \varphi^{\rm WL}(\hat\sfC^{\rm WL}\cdot s_\lam) \,.
\ee
\end{example}

\subsection{Cut-and-join or \texorpdfstring{$W$}{W}-representation}

In general, if a model is Schur superintegrable, we can write the generating function $\sfZ^V(\bp)$ as an exponential operator acting on the vacuum state, as follows
\be
\begin{aligned}
 \sfZ^V(\bp)
 &= \sum_\lam \varphi^V(\hat{\sfC}^V\cdot s_\lam) s_\lam(\bp) \\
 &= \hat{\sfC}^V \exp \left( \sum_{k\geq1} \frac{\varphi^V(p_k)p_k}{k} \right)
 (\hat{\sfC}^V)^{-1} \cdot 1 \\
 &= \exp \left( \sum_{k\geq1} \frac{\varphi^V_k}{k}
 \hat{\sfC}^V p_k (\hat{\sfC}^V)^{-1} \right) \cdot 1 \\
 &= \exp \left( \sum_{k\geq1} \frac{\varphi^V_k}{k}
 \hW^V_{-k} \right) \cdot 1\,,
\end{aligned}
\ee
where we defined the operators $\hW^V_{-k}:=\hat{\sfC}^V p_k (\hat{\sfC}^V)^{-1}$, which act on the Schur basis as
\be
 \hW^V_{-k}\cdot s_\lam = \sum_{\nu \supset \lam} \frac{\sfC^V_\nu}{\sfC^V_\lam}\,
 d^{(k)}_{\nu/\lam}\, s_\nu\,,
 \hspace{30pt}
 d^{(k)}_{\nu/\lam} := \langle s_\nu,p_k\, s_\lam\rangle\,.
\ee
Notice that, since $p_k$ form a commutative algebra and conjugation by the diagonal operator $\hat{\sfC}^V$ is an automorphism, it follows that all the operators $\hW^V_{-k}$ commute with each other (for a given choice of $V$).
We note also that the operators $\hW^V_{-k}$ can be regarded as elements of a representation of the $W_{1+\infty}$ algebra in a single free boson Fock space \cite{Awata:1994tf,Mironov:2020pcd}.

\section{\texorpdfstring{$q,t$}{q,t}-deformation}
\label{sec:qt-deformation}

We are now ready to introduce the $q,t$-deformed matrix models, which are the main objects of interest in what follows. It turns out that many properties of the classical Hermitian models can be extended to the $q,t$-deformed case.

\subsection{Definition of the models}\label{sec:Definitionqt}
The main ingredient of deformed matrix model is the substitution of the Vandermonde determinant with the function
\be
 \Delta_{q,t}(\bx) =
 \prod_{1\leq i\neq j\leq N}\frac{(x_i/x_j;q)_\infty}{(tx_i/x_j;q)_\infty}\,,
\ee
which plays the role of Macdonald kernel as in \eqref{eq:macdonald-kernel}.
Observe that for $t=q$, the ratio of $q$-Pochhammers simplifies drastically and one recovers the usual Vandermonde determinant,
\be
 \Delta_{q,q}(\bx) = \prod_{i\neq j}(1-x_i/x_j)
 = \prod_{i\neq j}(x_i-x_j) \prod_{i=1}^N x_i^{-N+1} \,.
\ee

We denote the potential (weight function) by $w(x)$, with the identification $w(x)=\mathe^{-V(x)}$. We use this notation because, in the $q,t$-deformed case, the potential appears less natural in its exponentiated form. This also allows us to distinguish the two cases more clearly. Accordingly, we write the normalized expectation values as
\be\label{eq:evw}
 \ev{f}^{w} = \dfrac{1}{\sfZ^w}
 \oint_{\mathcal{C}} \prod_{i=1}^{N} \mathd x_i\, \Delta_{q,t}(\bx) \, f(\bx)
 \prod_{i=1}^N x_i^{\frac{\log t}{\log q}(N-1)} w(x_i)\,.
\ee
Here $w(x)$ is part of the measure and characterizes the specific model while $\sfZ^w$ is simply a normalization factor corresponding to the partition function, i.e.\ the integral without any observable insertion:
\begin{equation}
    \sfZ^w=\oint_{\mathcal{C}} \prod_{i=1}^{N} \mathd x_i\, \Delta_{q,t}(\bx) \,
 \prod_{i=1}^N x_i^{\frac{\log t}{\log q}(N-1)} w(x_i)\,.
\end{equation}
As before, this definition ensures that $\ev{1}^w=1$.

The choice of the integration contour $\mathcal{C}$ depends on the specific model under consideration and is generally a subtle and nontrivial problem. In order to be able to derive $q,t$-Virasoro constraints later, we will assume that the contour satisfies the following condition: if a pole of the integrand is enclosed by $\mathcal{C}$, then all the $q$-shifted images of the pole are also enclosed by $\mathcal{C}$. When this is the case, the insertion of a total $q$-derivative in the integral is identically zero, which leads to a $q$-deformation of Stokes' theorem. This is a key ingredient in the derivation of $q$-deformed Ward identities.
We will return to the question of contours, when we define the respective models in the subsections below. We would like mention here that in the specific models that we will consider later, there will always be a unique way to make this choice.

\begin{remark}
In general, it is possible to rescale the weight function $w(x)$ by an arbitrary $q$-constant function $c_q(x)$, i.e.\
\be
 \tilde{w}(x) = c_q(x) w(x)\,,
 \hspace{30pt}
 D_x^q c_q(x)   = 0 \,,
\ee
so that the matrix integrals associated with the weight $\tilde{w}(x)$ obey the same system of difference equations as those defined with $w(x)$.
This allows to make contact with another well-known formulation of $q,t$-deformed matrix models, via Jackson integration. Namely, by choosing the $q$-constant and the contour in a suitable way, one can bring the original model into an $N$-fold Jackson integral form:
\be
 \ev{f}^{\tilde{w}}
 = \ev{f}^{w}_{\mathrm{Jackson}}
 = \dfrac{ \displaystyle\int_{a}^{b} \prod\limits_{i=1}^{N}\mathd_q x_i \,\Delta_{q,t}(\bx) f(\bx)
 \prod\limits_{i=1}^N x_i^{\frac{\log t}{\log q}(N-1)} w(x_i)}
 {\displaystyle\int_{a}^{b} \prod\limits_{i=1}^{N}\mathd_q x_i \,\Delta_{q,t}(\bx)
 \prod\limits_{i=1}^N x_i^{\frac{\log t}{\log q}(N-1)} w(x_i)} \,.
\ee
See the discussion in \cite{Lodin:2018lbz} and in Appendix~\ref{sec:AppendixContours} for the proof. We will later see that for the models in question, the choice of the $q$-constant function does not affect normalized expectation values, so that we can drop the distinction between $w$ and $\tilde{w}$ whenever the appropriate choice of integration is understood.
\end{remark}

Let us now specify more precisely the class of integrals under consideration.

\subsubsection{Supersymmetric gauge theories on \texorpdfstring{$D^2\times_q S^1$}{D^2xS^1}}

There exist a family of $q,t$-deformed matrix integrals associated to the weight function
\begin{equation}\label{eq:GeneralNfmodel}
 w(x) = x^{\frac{\log r}{\log q}}\,
 \exp\left(-\ell\frac{\log^2 x}{2\log q}\right)
 \prod_{a=1}^{\Nf} (qu_ax;q)_{\infty}
 \prod_{a=1}^{\Nbf} (qv_ax;q)^{-1}_{\infty} \,.
\end{equation}
These integrals arise as a result of supersymmetric localization of $3d$ $\mathcal{N}=2$ gauge theories on a $D^2 \times_q S^1$ background where the parameter $q$ corresponds to the omega-background deformation on the disk \cite{Beem:2012mb,Yoshida:2014ssa}. The content of these theories is given by a $U(N)$ gauge vector field together with an adjoint massive chiral multiplet plus $\Nf$ fundamental and $\Nbf$ antifundamental chiral matter fields.
The parameter $r$ denotes the exponentiated FI coupling for the central $U(1)$ factor of the gauge group and $\ell/2$ denotes the effective Chern--Simons level of $U(N)$.
The dictionary between the parameters of the matrix model and gauge theory is summarized in Table~\ref{tab:dictionary}.
\begin{table}[!ht]
    \centering
    \begin{tabular}{c|c}
        Omega background & $q$ \\
        Adjoint mass & $t$ \\
        Fundamental masses & $u_a$ \\
        Anti-fundamental masses & $v_a$ \\
        Fayet--Iliopoulos parameter & $r$ \\
        Chern--Simons level & $\ell$
    \end{tabular}
    \vspace{10pt}
    \caption{Dictionary between $3d$ gauge theory variables and matrix model parameters.}
    \label{tab:dictionary}
\end{table}

Averages of Macdonald polynomials in such matrix models compute expectation values of $1/2$-BPS Wilson loops wrapped around $\{0\}\times S^1$ where $\{0\}$ is the center of the 2-disk.

These integrals also play a role in the so called $A_n$-triality of \cite{Aganagic:2013tta,Aganagic:2014oia}, and because of this they also correspond to the Dotsenko--Fateev formulation of $q$-Liouville conformal blocks which are themselves related to $5d$ $\mathcal{N}=1$ Nekrasov functions via the AGT correspondence.

While the associated $3d$ gauge theories are well-defined for arbitrary values of $\Nf$ and $\Nbf$, we will restrict our focus to cases where $\max(\Nf,\Nbf)\leq2$. For values of $\Nf,\Nbf$ higher than two the recursion obtained through the $q,t$-Virasoro constraints is such that it requires more initial data and hence have a higher-dimensional space of solutions.
Similarly, we restrict to $|\ell|\leq2$ for the same reason.
We will therefore consider the following cases.

For $\Nf=1$ and $\Nbf=0$, we have
\be
\ba
 \ev{f}^{1\bar{0}} &= \frac{1}{\sfZ^{1\bar{0}}} \oint_\mathcal{C} \prod_{i=1}^N \mathd x_i\,
 \Delta_{q,t}(\bx) f(\bx) \prod_{i=1}^{N} \gamma_q(x_i|(qu_1)^{-1})\,
 x_i^{\frac{\log t}{\log q}(N-1)+\frac{\log r}{\log q}}
 (qu_1x_i;q)_{\infty} \\
 &= \frac{1}{\sfZ^{1\bar{0}}}\int\limits_{0}^{(qu_1)^{-1}} \prod_{i=1}^N \mathd_q x_i\,
 \Delta_{q,t}(\bx) f(\bx) \prod_{i=1}^N x_i^{\frac{\log t}{\log q}(N-1)+\frac{\log r}{\log q}}
 (qu_1x_i;q)_{\infty}
\ea
\ee
where $\gamma_q(x|a)$ is a $q$-constant function defined in \eqref{eq:def-gamma_q} and the contour of integration $\mathcal{C}$ is chosen such that it encloses some the poles of this function (see Appendix~\ref{sec:AppendixContours}).
In the latter Jackson integral presentation, the model can be regarded as a $q,t$-deformation of the Wishart--Laguerre ensemble. We could also turn on a Chern-Simons level. The same condition that restricts the number of matter contributions, restricts the possible Chern-Simons level to be $\ell=-1$, promoting the potential to be:
\begin{equation}
w(x) = x^{\frac{\log r}{\log q}}(q u_1x;q)_{\infty} \exp \left( \dfrac{\log^2x}{2 \log q} \right)\,.
\end{equation}
This modification does not alter the contour prescription.

For $\Nf=2$ and $\Nbf=0$, we have
\be
\begin{aligned}
 \ev{f}^{2\bar{0}} &= \frac1{\sfZ^{2\bar{0}}} \oint_{\mathcal{C}}\prod_{i=1}^N
 \mathd x_i\, \Delta_{q,t}(\bx)
 f(\bx) \prod_{i=1}^{N} \gamma_q(x_i|(qu_1)^{-1},(qu_2)^{-1})
 x_i^{\frac{\log t}{\log q}(N-1)+\frac{\log r}{\log q}}
 (qu_1x_i;q)_{\infty}(qu_2x_i;q)_{\infty} \\
 &= \frac1{\sfZ^{2\bar{0}}} \int_{(qu_1)^{-1}}^{(qu_2)^{-1}} \prod_{i=1}^N \mathd_q x_i\,
 \Delta_{q,t}(\bx) f(\bx) \prod_{i=1}^N x_i^{\frac{\log t}{\log q}(N-1)+\frac{\log r}{\log q}}
 (qu_1x_i;q)_{\infty}(qu_2x_i;q)_{\infty}\,.
\end{aligned}
\ee
with $\gamma_q(x|a,b)$ defined in \eqref{eq:def-gamma_q}.
An important special case corresponds to $u_1=-u_2$ and $r=q$, when this model is known as the $q,t$-Gaussian model \cite{Morozov:2018eiq,Mironov:2022fsr}.
Note that due to the normalization factors included in $\gamma_q(x|a)$ and $\gamma_q(x|a,b)$, the constants $\sfZ^{2\bar{0}}$ are the same for both integrals.

We treat the case of models with fundamental and anti-fundamental matter separately as some of our techniques work somewhat differently.

For $\Nf=\Nbf=1$, we have
\be
\ba
 \ev{f}^{1\bar{1}} &= \frac{1}{\sfZ^{1\bar1}} \oint_\mathcal{C} \prod_{i=1}^N
 \mathd x_i\, \Delta_{q,t}(\bx) f(\bx) \prod_{i=1}^{N}
 x_i^{\frac{\log t}{\log q}(N-1)+\frac{\log r}{\log q}}
 \frac{(qu_1x_i;q)_{\infty}}{(qv_1x_i;q)_{\infty}}\\
 &= \frac{1}{\tilde{\sfZ}^{1\bar1}} \int_{0}^{(qu_1)^{-1}} \prod_{i=1}^N \mathd_q x_i\,
 \Delta_{q,t}(\bx) f(\bx) \prod_{i=1}^N
 x_i^{\frac{\log t}{\log q}(N-1)+\frac{\log r}{\log q}}
 \frac{(qu_1x_i;q)_{\infty}}{(qv_1x_i;q)_{\infty}}
\ea
\ee
We assume that $u_1 \neq v_1$. This type of integrals are also known as the $q$-Selberg or $q$-Kadell integrals \cite{kaneko1996q}. The contour $\mathcal{C}$ is chosen in such a way, that it encloses the poles of the denominator, see a more detailed discussed in \cite{Aganagic:2012ne}. Note, that the normalization factors are different in this case. The equality of the two expectation values does not follow directly from contour analysis, but rather relies on the uniqueness to the solutions of $q,t$-Virasoro constraints. Namely, if one inserts a $\gamma_q(x|(qu_1)^{-1})$ function in the first integral, then one can equate the result to the second line directly. The fact that this insertion does not affect the normalized expectation values relies on the corresponding property of the $q,t$-Virasoro constraints. 
    
For $\Nf=\Nbf=2$, we have
\be
\ba
 \ev{f}^{2\bar{2}} &= \frac{1}{\sfZ^{2\bar2}} \oint_\mathcal{C} \prod_{i=1}^N
 \mathd x_i\, \Delta_{q,t}(\bx) f(\bx) \prod_{i=1}^{N}
 x_i^{\frac{\log t}{\log q}(N-1)+\frac{\log r}{\log q}} \frac{(qu_1x_i;q)_{\infty}(qu_2x_i;q)_{\infty}}
 {(qv_1x_i;q)_{\infty}(qv_2x_i;q)_{\infty}} \\
 &= \frac{1}{\tilde{\sfZ}^{2\bar2}} \int_{(qu_2)^{-1}}^{(qu_1)^{-1}} \prod_{i=1}^N \mathd_q x_i\,
 \Delta_{q,t}(\bx) f(\bx) \prod_{i=1}^{N} x_i^{\frac{\log t}{\log q}(N-1)+\frac{\log r}{\log q}}
 \frac{(qu_1x_i;q)_{\infty}(qu_2x_i;q)_{\infty}}{(qv_1x_i;q)_{\infty}(qv_2x_i;q)_{\infty}}
\ea
\ee
We assume that $u_i \neq v_j$, and the contour $\mathcal{C}$ is chosen in such a way to enclose all poles of the denominator.
The Jackson integral version of the model is also known as the $q$-Penner model.
Note that the normalization factors in the contour and Jackson integrals are different.

We briefly comment on the contour prescription in this model.
For $r\neq 1$, the contour integral is not uniquely defined, as discussed in \cite{Aganagic:2013tta,Aganagic:2014kja}. In this case the model realizes the full four-point $q$-Liouville conformal block. When $r=1$, however, one of the vertex–operator insertions effectively decouples, so the model reduces essentially to a three-point function.

We can instead set the fundamental masses to zero, hence obtaining two other models. First, we have $\Nf=0$ and $\Nbf=1$, however, just as the $1\bar{0}$ model allowed for a Chern-Simons level, to this model we can add an opposite $\ell=1$ Chern-Simons term, giving:
\begin{equation}
 w(x)= x^{\frac{\log r}{\log q}}\dfrac{\exp\left( - \dfrac{\log^2 x}{2\log q} \right)}{(q v_1x;q)_{\infty}}
\end{equation}
Finally, we have the $\Nf=0$ and $\Nbf=2$ case with the weight function
\begin{equation}
 w^{0\bar{2}}(x)= \dfrac{1}{(q v_1x;q)_{\infty}(q v_2x;q)_{\infty}}
\end{equation}
As for the contour prescription, it is clear that now, only the contour, that goes around residues of the denominator makes sense, and we do not have a Jackson integral representation for these models. 

Finally, we note that in the models discussed above, one could in principle choose alternative integration contours---for instance, contours enclosing only a subset of the poles of the integrand. Such choices, however, do not satisfy the $q$-shift condition on the poles introduced at the beginning of this section, and therefore the methods developed in this paper are not applicable.

\subsubsection{Refined Chern--Simons theory on \texorpdfstring{$S^3$}{S^3}}
The matrix model for the refined Chern--Simons theory is given by the weight function
\begin{equation}
 w(x) = x^{\frac{\log r}{\log q}}\, \exp\left(-\frac{\log^2 x}{2\log q}\right)\,,
\end{equation}
which is a special case of \eqref{eq:GeneralNfmodel} for $\Nf=\Nbf=0$ and $\ell=1$.
This particular case is somewhat special as it can be regarded as a $3d$ gauge theory in two different and unrelated ways: namely, one can think of it as a subcase of the $\mathcal{N}=2$ gauge theories discussed in the previous section, but also as a ``refinement'' of the topological Chern--Simons theory for $U(N)$ on the three-sphere.
As explained in \cite{Aganagic:2012ne,Aganagic:2011sg}, this is a one-parameter deformation of the usual CS matrix model. The parameter $q$ takes on a different role in this picture, in fact, it is identified with the exponentiated CS coupling as,
\begin{equation}
 q = \mathe^{\frac{2\pi\mathi}{k+\beta N} } \,,
 \hspace{30pt}
 t=q^{\beta}\,,
\end{equation}
where $k$ is the level and $\beta$ is the refinement/deformation parameter.\footnote{For $\beta=1$ and $r=1$, one recovers the classical Stieltjes--Wigert ensemble \cite{Marino:2002fk,Tierz:2002jj,Aleksandrov:2014jvm}.} As it is clear for this discussion, both $\ell$ and $k$ can be regarded as CS levels but not for the same theory, and in fact they are independent of each other.

The rCS matrix model average of a Macdonald polynomial, $\ev{P_\lam}^\rCS$, computes the expectation value of a Wilson loop wrapping around an unknot inside of $S^3$ for the representation of $U(N)$ labeled by $\lam$ (at framing 1).
As argued in \cite{Aganagic:2012ne,Aganagic:2011sg}, this quantity is related to the superpolynomial of the unknot in the same representation $\lam$ (see also \cite{Cherednik:2011nr}).

Due to a theorem of Cherednik \cite{Cherednik:1997,Etingof:1997} also known as Cherednik--Macdonald--Mehta identity, expectation values of products of Macdonald polynomials admit a closed formula given by
\be\label{eq:CMM}
 \ev{ P_\lam \, P_\mu }^{\rCS}
 = \sfT_\lam\, P_\lam(u_\vac)\, P_\mu(u_\lam)\, \sfT_\mu
\ee
where $\sfT_\lam$ and $u_\lam$ are defined as in \eqref{eq:symbols1} and \eqref{eq:symbols3}, respectively. This quantity is symmetric under the exchange of $\lam$ and $\mu$ due to the Macdonald--Koornwinder duality, and it is known to compute the superpolynomial of an Hopf link in $S^3$.

\begin{remark}
Choosing $\ell=-1$, instead, we obtain a slightly different version on refined CS theory. At the level of the weight function, this corresponds to the operation of sending $q$ to $q^{-1}$, which for the unrefined theory is known to be related to a change of orientation on the three-sphere. When $t\neq q$ however, the kernel $\Delta_{q,t}(\bx)$ is not invariant under the inversion of the parameters,
\be
 \Delta_{q^{-1},t^{-1}}(\bx) = \prod_{i\neq j}\frac{((q/t)x_i/x_j;q)_\infty}
 {(q x_i/x_j;q)_\infty} \neq \Delta_{q,t}(\bx)\,.
\ee
Nevertheless, we will refer to this case as rCS with opposite orientation.
\end{remark}

\subsection{\texorpdfstring{$q,t$}{q,t}-Virasoro constraints}
Similarly to the case of classical Hermitian matrix models, one can derive a set of linear identities satisfied by the correlation functions of the $q,t$-deformed ensembles. These correspond to the Ward identities of the matrix model and they are derived by inserting certain total $q$-difference operators in a way which is analogous to \eqref{eq:Virasoro1}.
For an appropriate choice of difference operators, one is able to derive constraints which realize a $q$-difference D-module for the deformed Virasoro algebra of \cite{Shiraishi:1995rp}, as shown in \cite{Zenkevich:2014lca,Mironov:2016yue,Lodin:2018lbz}.

Here we rely on these results to derive the general form of $q,t$-Virasoro constraints for a model with an arbitrary weight function $w(x)$, and then specify to the models defined in the previous section.

The existence of the constraints relies on the following property of the $q$-derivative.
For a contour $\mathcal{C}$ and a function $f(z)$ such that, if for any pole of $f(z)$ that sits inside of $\mathcal{C}$, the corresponding pole of $f(qz)$ is also inside of $\mathcal{C}$, then
\be
 \oint_\mathcal{C} \mathd x\, D_x^{q}f(x)=0\,.
\ee
The same idea can also be applied to Jackson integrals. In particular,
\be
 \int_a^b \mathd_q x\, D_x^{q}f(x) = f(b)-f(a)
\ee
so that, choosing $a$ and $b$ such that $f(a)=f(b)$, one obtains a vanishing integral.

We define the generating function of the expectation values of the $q,t$-deformed matrix model with weight function $w(x)$, as
\be\label{eq:gen-func-qt}
\ba
 \sfZ^w(\bp) :
 =&\, \frac{1}{\sfZ^w}\oint \prod_{i=1}^N \mathd x_i\, \Delta_{q,t}(\bx)
 \prod_{i=1}^{N} x_{i}^{\frac{\log t}{\log q}(N-1)} w(x_i)
 \exp\Big( \sum_{k\geq1} \dfrac{1-t^k}{1-q^k} \dfrac{p_k\, x_i^k}{k} \Big) \\
 =&\, \ev{ \exp\Big( \sum_{k\geq1}
 \frac{1-t^k}{1-q^k} \frac{p_k}{k} \sum_{i=1}^N x_i^k \Big) }^w
 = \sum_\lam b_\lam \ev{ P_\lam(\bx) }^w P_\lam(\bp) \,.
\ea
\ee

\begin{lemma}
\label{lemma:qtVirUniv}
Suppose there exist two non-zero polynomials $M^+(z)$ and $M^-(z)$,
\be
 M^{\pm}(z) := \sum_{k=0}^{d_{\pm}} M^{\pm}_k z^k
\ee
such that $\max(d_{+},d_{-})\leq2$ and
\be\label{eq:mutoM}
 \frac{w(z)}{w(q^{-1}z)} = \frac{M^+(z)}{M^-(z)}\,.
\ee
Then the generating function of the $q,t$-deformed matrix model with weight function $w(x)$ satisfies the equations
\be
 \hat{\sfU}_m^{w} \cdot \sfZ^{w}(\bp) = 0 \,,
 \hspace{30pt}
 m \geq m_0 \,,
\ee
where $\hat{\sfU}^w_m$ are the modes of the generating current $\hat{\sfU}^w(z)=\sum\limits_{m\in\BZ}\hat{\sfU}^w_m z^{-m}$, defined as
\begin{multline}
\label{eq:U(z)}
 \hat{\sfU}^w (z) = M^-(z) \left\{ 1-Q^{-1} \exp\left( \sum_{k\geq1}
 (1-q^k)z^{-k} \frac{\partial}{\partial p_k} \right) \right\} \\
 + (qt^{-1}) M^+(z) \exp\left( \sum_{k\geq1} (1-t^{-k})(t/q)^{k}z^{k}\frac{p_k}{k}\right)
 \left\{ 1 - Q \exp\left( - \sum_{k\geq1} (1-q^{k})(t/q)^{-k}z^{-k}
 \frac{\partial}{\partial p_k} \right) \right\}\,,
\end{multline}
and we introduced $Q:=t^N$.

For generic $r$, the lower bound is $m_0=0$, i.e.\ only the non-negative modes of $\hat{\sfU}^w(z)$ annihilate the generating function.

For $M^{+}(0)=M^{-}(0)\neq0$, we have an additional equation and hence $m_0=-1$. 
\end{lemma}
\begin{proof}
We prove the statement by deriving these equations explicitly via $q$-derivative insertions, following \cite{Lodin:2018lbz,Cassia:2020uxy,Cassia:2021uly}.
First, introduce the $q$-difference operators
\be
 \mathcal{D}_m\, f(\bx)
 = \sum_{i=1}^N D_{x_i}^{q}\left(x_i^{m+1}\,\prod_{j \neq i} \frac{t^{-1} x_i-x_j}{x_i-x_j}
 f(\bx) \right)
\ee
with generating current $\mathcal{D}(z):=\sum\limits_{m\in\BZ}\mathcal{D}_mz^{-m}$.
Inserting this current under the integral in \eqref{eq:gen-func-qt}, we obtain
\be
 0 = \oint\prod_{i=1}^N\mathd x_i \,  \sum_{m\in \BZ} z^{-m}
 \mathcal{D}_m\left\{\Delta_{q,t}(\bx) \,
 \prod_{i=1}^N x_i^{\frac{\log t}{\log q}(N-1)} w(x_i)
 \exp\left( \sum_{k\geq1} \frac{1-t^k}{1-q^k}
 \frac{p_k\,x_i^k}{k} \right) \right\}\,.
\ee
We now wish to act with the operators $\mathcal{D}_m$ inside of the integrand and rewrite the result as operators in the variables $p_k$ acting on $\sfZ^w(\bp)$. To present the result of this action, we introduce the shorthand notation
\be
 {\dev{f(\bx)}^w} := \ev{f(\bx) \exp\Big( \sum_{k\geq1}
 \frac{1-t^k}{1-q^k} \frac{p_k}{k}\sum_{i=1}^N x_i^k\Big)}^w
 = f^\perp\cdot \sfZ^w(\bp)\,.
\ee
We can then rewrite the vanishing of the total $q$-derivative insertion as the constraint equation
\be
\label{eq:GenericVirasoroAll}
\ba
 & \dev{ \exp\Big( \sum_{k\geq1} (1-t^{-k}) \frac{z^k}{k}
 \sum_{i=1}^N x_i^{-k}\Big) }^w
 -Q^{-1}\exp\Big( \sum_{k\geq1} (1-q^k) z^{-k} \frac{\partial}{\partial p_k}\Big)
 \sfZ^w(\bp) \\
 & + (q t^{-1})\frac{w(z)}{w(q^{-1}z)} \exp\Big( -\sum_{k\geq1} (1-t^k)q^{-k} z^k
 \frac{p_k}{k} \Big) \times \\
 &\times \left\{ \dev{ \exp\Big( -\sum_{k\geq1} (1-t^{-k})(t/q)^{k} \frac{z^k}{k}
 \sum_{i=1}^N x_i^{-k} \Big)}^w \right.\\
 &- Q \left.\exp\Big( -\sum_{k\geq1}
 (1-q^{k})(t/q)^{-k} z^{-k} \frac{\partial}{\partial p_k} \Big) \sfZ^w (\bp) \right\} =0\,.
\ea
\ee
Next, we use \eqref{eq:mutoM} and multiply the whole equation by $M^-(z)$.
The terms in the equation that contain negative power sums $p_{-k}(\bx)=p_k(\bx^{-1})$ cannot be expressed in terms of differential/difference operators in $p_k$ and therefore they either cancel out or they do not lead to equations for the partition function. Therefore, for generic $r$, we consider only the non-positive powers of $z$ in the series expansion of this equation. The positive powers will contain negative power sums. The coefficient of $z^1$ is somewhat special as it contains one term proportional to $\dev{p_{-1}(\bx)}^w$ which has the form
\be
 \oint\frac{\mathd z}{2\pi\mathi z}z^{-1} M^-(z) \left(\ref{eq:GenericVirasoroAll} \right)
 = (1-t^{-1}) \left( M^-(0) - M^+(0) \right){\dev{p_{-1} (\bx)}^w} +\dots 
\ee
Therefore, when $M^{+}(0)=M^{-}(0)\neq0$, this term vanishes, while the rest of this equation results in a well-defined differential operator acting on $\sfZ^w(\bp)$.
This situation typically, happens when $\max(d_{+},d_{-})=2$ and $\ell=0$, in which case $M^{+}(0)=r$ and $M^{-}(0)=1$.
Hence, if we fix $r=1$, we can still safely define the $m=-1$ constraint and substitute the terms of the form $\dev{\dots}^w$ in \eqref{eq:GenericVirasoroAll} with 1. Doing so, we obtain the operator $\hat{\sfU}^w(z)$ in \eqref{eq:U(z)}.
\end{proof}
The Virasoro constraints for generic models of type \eqref{eq:GeneralNfmodel} were studied in \cite{Lodin:2018lbz}. For $\Nf=1,2$, they where worked out in more detail in \cite{Cassia:2020uxy} and for the refined CS model in \cite{Cassia:2021uly}. For the $q$-Selberg integral these constraints where discussed in a slightly different form in \cite{kaneko1996q,Zenkevich:2014lca}. See the discussion in \cite{Lodin:2018lbz} for the relation between the operators above and the generators of the $q,t$-Virasoro algebra, and hence with the approaches of \cite{Zenkevich:2014lca,Awata:2016riz}.

As we can see from \eqref{eq:U(z)}, the $q,t$-Virasoro constraints are invariant under a redefinition of $w(x)$ by a multiplicative $q$-constant function, since this contribution to the integrand simply cancels in the ratio $\frac{w(z)}{w(q^{-1}z)}$. Then we have the following.
\begin{corollary}
Let us consider two matrix models with weight functions $w_1(x)$ and $w_2(x)$, respectively, such that $\frac{w_1(x)}{w_2(x)}=c_q(x)$ is some $q$-constant function. Then their generating functions satisfy the same $q,t$-Virasoro constraints, so that
\be
 \hat{\sfU}^{w_1}(z) = \hat{\sfU}^{w_2}(z)\,.
\ee
Moreover, if the space of solutions of these constraints is one dimensional, then these generating functions are equal,
\be
 \sfZ^{w_1}(\bp) = \sfZ^{w_2}(\bp) \,. 
\ee
\end{corollary}
In the models considered in this paper, this situation indeed occurs. This corollary confirms the validity of the definitions introduced in the previous subsection, as it establishes that the Virasoro constraints in these models admit a unique solution. In particular, as noted above, replacing a contour integral with a Jackson integral entails modifying the weight function by a $q$-constant function. Since such a modification is absorbed by the normalization, the normalized expectation values remain unchanged.
\begin{proposition}
Suppose the potential $w(z)$ is such that  $M^-(x)$ and $M^+(x)$ are polynomials of degree at most $2$ and such that $M^{+}(0)M^{-}(0) \neq 0$ when $\max(d_{+},d_{-})=2$. Then the $q,t$-Virasoro constraints have a unique solution in the space of formal power series in $p_k$ with constant term equal to $1$.
\end{proposition}
\begin{proof}
We will prove this statement case by case in the next section, by deriving explicitly the unique solution.
\end{proof}
Note that for higher–degree polynomials, the solution space of the $q,t$-Virasoro constraints appears to become higher-dimensional. This behaviour is analogous to that observed in higher monomial matrix models. A systematic treatment of such integrals is left for future work.

\section{Superintegrability from recursions}
\label{sec:macdonald-SI}

In this section we prove superintegrability for the $q,t$-models introduced in the preceding section. Our strategy is to reformulate the problem in a purely algebraic framework. By combining the $q,t$-Virasoro constraints in an appropriate way, we are able to exploit specific properties of the quantum toroidal algebra $\DIM$ to find closed formulas for the solutions.

Superintegrability of $q,t$-deformed matrix models was first conjectured in some examples in \cite{Morozov:2018eiq,Cassia:2020uxy} and recently proved in the case of the $q,t$-Gaussian model using orthogonal polynomials techniques in \cite{Byun:2025qrv}. Superintegrability of the rCS matrix model follows instead from earlier work of Cherednik on the Macdonald--Mehta identities \cite{Cherednik:1997,Etingof:1997}. Here we give an alternative proof based on the deformed Virasoro constraints, as well as new proofs for those cases which have not been considered before in the literature.

Just as in the case of classical matrix models, we first need to provide a definition. \begin{definition}[Macdonald superintegrability of refined matrix models]\label{def:Macdonald-SI}
A $q,t$-matrix model is said to be Macdonald superintegrable if there exist an operator $\hat\sfC^w:\Lambda\to\Lambda$ diagonal on Macdonald functions,
\be
 \hat\sfC^w\cdot P_\lam = \sfC_\lam^w \, P_\lam\,,
 \hspace{30pt}
 \text{with eigenvalues}
 \hspace{30pt}
 \sfC_\lam^w = \prod_{(i,j) \in \lam} g^w(i,j)\,,
\ee
for some function $g^w(i,j)$ of the coordinates of the boxes in the partition, which depends on the choice of potential $w$, and a ring homomorphism $\varphi^w: \Lambda \to \BC$, that acts as a specialization of the power sums to some $w$-dependent locus such that
\be
 \ev{ f }^w = \varphi^w \left(\hat\sfC^w \cdot f(\bp)\right)\,,
\ee
for any symmetric function $f\in\Lambda$.
\end{definition}
In particular, this means that the averages of Macdonald
polynomials in a superintegrable $q,t$-model can be written as
\be
 \ev{P_\lam}^w
 = \varphi^w \left(\hat\sfC^w \cdot P_\lam(\bp)\right)
 = \sfC^w_\lam\, P_\lam (p_k = \varphi^w_k)\,,
\ee
where we denoted as $\varphi^w_k:=\varphi^w(p_k)$ the evaluation of the $k$-th power sum.
Equivalently, if we define the normalized Macdonald polynomials
\be
 \tilde{P}^w_\lam(\bx) := \frac{P_\lam(\bx)}{P_\lam(p_k=\varphi^w_k)}\,,
\ee
then superintegrability becomes
\be
 \ev{ \tilde{P}^w_\lam(\bx) }^w = \sfC^w_\lam\,,
\ee
which one should regard as a version of a polynomial difference Fourier transform in the sense of \cite{Cherednik:1995mac,Stokman2000:koo,Stokman2001:dif}. In particular, we regard $\tilde{P}^w_\lam(\bx)$ as the (polynomial) kernel of the transform and the weight function $w(\bx)=\prod_i w(x_i)$ as the argument of the transform. The case of rCS, in this language, corresponds to the Fourier transform of the generalized Gaussian $w^\rCS(\bx)$.

\begin{remark}
In all known examples, the function $g^w(i,j)$ is a function of the content \eqref{eq:box-content} of the box only. The generating function $\sfZ^w(\bp)$ can then be regarded as a specialization of a weighted $q,t$-tau function of the type recently considered in \cite{Dali:2025pat}.
\end{remark}

With this definition at hand, we are now ready to state the main result of this paper.
\begin{theorem}
\label{thm:SIqt}
The following models feature Macdonald superintegrability with $\hat{\sfC}^w$ and $\varphi^w_k$ as listed:
\begin{enumerate}[label=\rm\textbf{(SI.\arabic*)}]

\item \label{it:SI:1} refined Chern--Simons, i.e.\ $(\Nf,\Nbf)=(0,0)$ and $\ell=1$:
\be\label{eq:SI:1}
 \makebox[\textwidth]{$
 w^{\rCS}(x) = x^{\frac{\log r}{\log q}}\,\exp\left(-\dfrac{\log^2 x}{2\log q}\right),
 \hfill
 \hat{\sfC}^{\rCS} = \sfT \, \Delta^{+}(Q),
 \hfill
 \varphi^{\rCS}_k= (q^\frac12 rqt^{-1}Q)^k \dfrac{1}{1-t^k};
 $}
\ee

\item \label{it:SI:2} refined Chern--Simons with opposite orientation, i.e.\ $(\Nf,\Nbf)=(0,0)$ and $\ell=-1$:
\be\label{eq:SI:2}
 \makebox[\textwidth]{$
 w^{\overline{\rCS}}(x) = x^{\frac{\log r}{\log q}}\,\exp\left(\dfrac{\log^2 x}{2\log q}\right),
 \hfill
 \hat{\sfC}^{\overline{\rCS}} = \sfT^{-1} \, \Delta^{+}(Q),
 \hfill
 \varphi^{\overline{\rCS}}_k= (q^\frac12 r^{-1} (qt^{-1}Q)^{-2})^k \dfrac{1}{1-t^k};
 $}
\ee

\item \label{it:SI:3} $(\Nf,\Nbf)=(1,0)$ and $\ell=0$:
\be\label{eq:SI:3}
 \makebox[\textwidth]{$
 w^{1\bar{0}}(x) = x^{\frac{\log r}{\log q}}\,(qu_1x;q)_\infty\,,
 \hfill
 \hat{\sfC}^{1\bar{0}} = \Delta^{+}(Q) \Delta^{+}(rqt^{-1}Q) \,, 
 \hfill
 \varphi^{1\bar{0}}_k = \dfrac{u_1^{-k}}{1-t^k} \,;
 $}
\ee

\item \label{it:SI:4} $(\Nf,\Nbf)=(0,1)$ and $\ell=0$:
\be\label{eq:SI:4}
 \makebox[\textwidth]{$
 w^{0\bar{1}}(x) = x^{\frac{\log r}{\log q}}\,(qv_1x;q)^{-1}_\infty\,,
 \hfill
 \hat{\sfC}^{0\bar{1}} = \sfT^{-2} \Delta^{+}(Q) \Delta^{+}(rqt^{-1}Q) \,, 
 \hfill
 \varphi^{0\bar{1}}_k = -(r^{-1}(qt^{-1}Q)^{-2})^k \dfrac{v_1^{-k}}{1-t^k} \,;
 $}
\ee

\item \label{it:SI:5} $(\Nf,\Nbf)=(2,0)$, $\ell=0$ and $r=1$:
\be\label{eq:SI:5}
 \makebox[\textwidth]{$
 w^{2\bar{0}}(x) = (qu_1x;q)_\infty(qu_2x;q)_\infty\,,
 \hfill
 \hat{\sfC}^{2\bar{0}} = \Delta^{+}(Q) \,,
 \hfill
 \varphi^{2\bar{0}}_k = \dfrac{(u_1^{-k}+u_2^{-k})}{1-t^k} \,;
 $}
\ee

\item \label{it:SI:6} $(\Nf,\Nbf)=(0,2)$, $\ell=0$ and $r=1$:
\be\label{eq:SI:6}
 \makebox[\textwidth]{$
 w^{0\bar{2}}(x) = (qv_1x;q)^{-1}_\infty(qv_2x;q)^{-1}_\infty\,,
 \hfill
 \hat{\sfC}^{0\bar{2}} = \sfT^{-1} \Delta^{+}(Q) \,,
 \hfill
 \varphi^{0\bar{2}}_k = -(-(qt^{-1}Q)^{-1})^k \dfrac{(v_1^{-k}+v_2^{-k})}{1-t^k} \,;
 $}
\ee

\item \label{it:SI:7} $(\Nf,\Nbf)=(1,1)$ and $\ell=0$:
\be\label{eq:SI:7}
 \makebox[\textwidth]{$
 w^{1\bar{1}}(x) = x^{\frac{\log r}{\log q}}\,\dfrac{(qu_1x;q)_\infty}{(qv_1x;q)_\infty}\,,
 \hfill
 \hat{\sfC}^{1\bar{1}} = \dfrac{ \Delta^{+}(Q) \Delta^{+}(rqt^{-1}Q) }
 {\Delta^{+}(r (q t^{-1}Q)^{2} \frac{v_1}{u_1})} \,,
 \hfill
 \varphi^{1\bar{1}}_k = \dfrac{u_1^{-k}}{1-t^k} \,;
 $}
\ee

\item \label{it:SI:8} $(\Nf,\Nbf)=(1,0)$ and $\ell=-1$:
\be\label{eq:SI:8}
 \makebox[\textwidth]{$
 w(x) = x^{\frac{\log r}{\log q}}\,\exp\left(\dfrac{\log^2 x}{2\log q}\right)
 (qu_1x;q)_\infty\,,
 \hfill
 \hat{\sfC}^w = \dfrac{\Delta^{+}(Q)}{\Delta^{+}(-q^{-\frac12} r (qt^{-1}Q)^2 u_1^{-1})} \,, 
 \hfill
 \varphi_k^w = \dfrac{u_1^{-k}}{1-t^k} \,;
 $}
\ee

\item \label{it:SI:9} $(\Nf,\Nbf)=(0,1)$ and $\ell=1$:
\be\label{eq:SI:9}
 \makebox[\textwidth]{$
 w(x) = x^{\frac{\log r}{\log q}}\,\dfrac{\exp\left(-\frac{\log^2 x}{2\log q}\right)}
 {(qv_1x;q)^{-1}_\infty}\,,
 \hfill
 \hat{\sfC}^w = \dfrac{\Delta^{+}(Q)}{\Delta^{+}(-q^\frac12 r (qt^{-1}Q)^2 v_1)} \,, 
 \hfill
 \varphi_k^w = (q^\frac12 r qt^{-1}Q)^k\dfrac{1}{1-t^k} \,;
 $}
\ee

\item \label{it:SI:10} $(\Nf,\Nbf)=(2,2)$, $\ell=0$ and $r=1$:
\be\label{eq:SI:10}
\begin{gathered}
 w^{2\bar{2}}(x) = \dfrac{(qu_1x;q)_\infty(qu_2x;q)_\infty}{(qv_1x;q)_\infty(qv_2x;q)_\infty}\,,
 \hspace{30pt}
 \hat{\sfC}^{2\bar{2}} = \dfrac{\Delta^{+}(Q)}
 {\Delta^{+}(\left({qt^{-1}Q}\right)^{2}\frac{v_1 v_2}{u_1 u_2})} \,, \\
 \varphi^{2\bar{2}}_k = \left( \left(u_1^{-k}+u_2^{-k}\right)-(v_1^{-k}+v_2^{-k})\left((q t^{-1}Q)\frac{v_1v_2}{u_1 u_2} \right)^{k}\right)\dfrac{1}{1-t^k} \,;
\end{gathered}
\ee

\end{enumerate}
where $\sfT$ and $\Delta^{+}(z)$ are defined as in \eqref{eq:framing} and \eqref{eq:DeltaOp}, respectively, and $Q=t^N$.
\end{theorem}

Notice that certain superintegrability formulas arise as special limits of others. In particular, we get
\begin{equation}
    \lim_{\substack{v_1 \to 0\\v_2 \to 0 }}  \ev{f}^{2\bar2} = \ev{f}^{2\bar{0}}\,,
\end{equation}
and 
\begin{equation}
    \lim_{\substack{v_1 \to 0}}  \ev{f}^{1\bar{1}} = \ev{f}^{1\bar{0}}\,.
\end{equation}
Because the model parameters are identified with the physical masses, this limit corresponds to the decoupling of anti-fundamental matter. One could analogously decouple the fundamental matter instead.
Crucially, however, the $\Nf=1$ model cannot be naively obtained as a limit of the $\Nf=2$ model. This obstruction is related to the change of integration contours relevant for these models, a topic lying beyond the scope of this paper. As we show below, both limits are nevertheless visible at the level of the Virasoro constraints and the corresponding equations for the generating functions.

The remainder of this section is devoted to the proof of Theorem~\ref{thm:SIqt}.

\subsection{Summing up Virasoro constraints}
Similarly to the case of the classical matrix models, the $q,t$-Virasoro constraints can be combined into a single equation, that can be used to determine the generating function uniquely. We start by presenting the general form of the resummed constraint equation. Next, we prove the announced results case by case. The proof uses the lemma below, which is one of the key technical result of this paper. It relies on a simple yet important idea, that one has to resum the Virasoro constraints, to bring them to a form, where they can be solved explicitly.
Yet, the resummation should be done in such a way that the resulting recursion operator has some ``\emph{nice}'' algebraic properties, as we explain below.
Resumming the constraints as in the undeformed case in Section~\ref{sec:Virasoro-constraints-classical}, for example, does not result in a particularly useful operator, as discussed in \cite{Cassia:2020uxy}.
\begin{lemma}
The matrix model generating function $\sfZ^w(\bp)$ satisfies the following equations for  $d_{-} \leq 1+m-m_0$, where $m_0$ is as in Lemma~\ref{lemma:qtVirUniv}:
\be\label{eq:RecUni}
\ba
 & \left( \sum_{l=0}^{d_{-}} M_{l}^{-} \Bigg\{
 (t/q)^{m-l} h_{m-l}( -\left(1-t^{-k} \right)p_k)
 - \delta_{l,m}
 - Q^{-1}\left( q^{l-m} x^{-}_{l-m}
 - h_{l-m}( (1-t^k) p_k^\perp) \right)
 \Bigg\}\right. \\
 & + q t^{-1} \sum_{l=0}^{d_{+}} M^+_l \Bigg\{\delta_{l,m}
 - (t/q)^{m-l} h_{m-l} ((1-t^{-k}) p_k) \\
 &\hspace{90pt}\left.+ Q (t/q)^{m-l}
 \left( x^{+}_{l-m} - h_{l-m}(-(1-t^k) p_k^\perp) \right)\Bigg\}
 \right) \cdot \sfZ^w(\bp) =0 
\ea
\ee
where $h_n$ are the homogeneous symmetric functions regarded as functions of power sums $p_k$, as in \eqref{eq:h_kdef}, and $p_k^\perp$ is the adjoint of $p_k$ w.r.t.\ the Macdonald inner product.
\end{lemma}
\begin{proof}
The proof is a direct computation.
Multiplying the constraints from the left, we obtain
\begin{equation}
 \oint \frac{\mathd z}{2\pi\mathi z} z^{-m}
 \left\{ \exp\left( -\sum_{k\geq1} (1-t^{-k})(t/q)^k z^k \frac{p_k}{k} \right)
 -1 \right\} \hat{\sfU}^w(z) \cdot \sfZ^w(\bp) = 0\,.
\end{equation}
According to Lemma~\ref{lemma:qtVirUniv}, the modes of $\hat{\sfU}^w(z)$ starting with $\hat{\sfU}_{m_0}^w$ annihilate the partition function. Since we are multiplying by an exponent that contains only positive powers of $z$, the resulting operator annihilates the partition function for  $m \leq 1-m_0$. This condition guarantees that only those modes of the $q$-Virasoro constraints, that annihilate the partition function appear in the resulting operator.

Next we substitute the explicit form of $\hat{\sfU}^w(z)$ into the expression and combine the exponents, some of which produce vertex operators for the currents of the quantum toroidal algebra. See Appendix~\ref{sec:AppendixDIM} for details. This produces the equation
\be
\ba
 \oint \frac{\mathd z}{2\pi\mathi z} z^{-m}
 \Big(&
 M^{-}(z)\Big\{
 \exp\left( -\sum_{k\geq1} (1-t^{-k})(t/q)^k z^k \frac{p_k}{k} \right)-1
 -Q^{-1}x^{-}(q^{-1}z) \\
 &+Q^{-1}\exp\left( \sum_{k\geq1} (1-q^{k}) z^{-k} \frac{\partial}{\partial p_k} \right)
 \Big\} \\
 &+qt^{-1} M^{+}(z)\Big\{
 1-Q\exp\left(-\sum_{k\geq1}(1-q^{k})(t/q)^{-k}z^{-k}\frac{\partial}{\partial p_k} \right)\\
 &-\exp\left( \sum_{k\geq1} (1-t^{-k})(t/q)^k z^k \frac{p_k}{k} \right)
 +Qx^{+}(tq^{-1}z)
 \Big\}\Big)\cdot\sfZ^w(\bp)=0
\ea
\ee
Taking the residue at zero extracts the $m$-th mode of the resulting operator and we obtain the equation in \eqref{eq:RecUni}.
\end{proof}
Out of all the equations for each $m$ we are going to be mainly interested in the lowest mode. Suppose for simplicity, that $d_- \geq d_+$, then, by taking $m=d_-$ we obtain the following.
\begin{corollary}
When $d_{-} \leq 1-m_0$, the generating function satisfies the linear equation
\be\label{eq:recursion-eq}
 \hat{\sfA}^w\cdot\sfZ^{w}(\bp) = 0 
\ee
where the recursion operator $\hat{\sfA}^w$ is given by
\be\label{cor:RecUni0}
\ba
 & \hat{\sfA}^w := \sum_{l=0}^{d_{-}}
 M^{-}_{d_{-}-l} \Big\{ Q^{-1} (\delta_{l,0}-q^{-l} x^{-}_{-l})
 + (t/q)^l h_l(-(1-t^{-k})p_k) - \delta_{l,0} \Big\} \\
 &- qt^{-1}
 \sum_{l=0}^{d_{+}}
 M^{+}_{d_{+}-l} \Big\{ Q (t/q)^{d_{-}-d_{+}+l}
 \left(\delta_{d_{+},d_{-}+l} - x^{+}_{d_{+}-d_{-}-l} \right) +
 h_{d_{-}-d_{+}+l}((1-t^{-k})(t/q)^k p_k) - \delta_{d_{+},d_{-}+l}
 \Big\}\,.
\ea
\ee
\end{corollary}
This equation is one of the main results of this work. Below we discuss how to solve this equation explicitly once we specialize to the models defined in Section~\ref{sec:Definitionqt}. The condition $d_{-} \leq 1-m_0$ appears, since we had $m \leq 1-m_0$ and translates to either $d_{-} \leq 1 $ for generic $r$ or $d_{-} \leq 2 $, when $r=1$. This is another manifestation of the fact, that for potentials with higher $d_{-}$ our approach should be modified and is the same issue as the one covered in \cite{Cassia:2020uxy}.

Similarly, if $d_{+}\geq d_{-}$, we can fix $m=d_{+}$ in \eqref{eq:RecUni} and we obtain an analogous recursion operator, provided $d_{+} \leq 1-m_0$.

\begin{remark}
Since the recursion equation \eqref{eq:recursion-eq} is homogeneous, the recursion operator $\hat{\sfA}^w$ is only defined up to an overall constant, which is in fact equivalent to the fact that the polynomials $M^{\pm}(z)$ can be simultaneously rescaled by the same scalar factor without changing the $q,t$-Virasoro constraints in Lemma~\eqref{lemma:qtVirUniv}.
In what follows, the operators are presented in rescaled form, with the rescaling absorbed into their definitions.
\end{remark}

\subsection{Refined Chern--Simons theory}\label{sec:CS}

\subsubsection{Proof of \ref{it:SI:1}}
The refined Chern--Simons model is arguably the simplest example, which we will discuss in more detail. The techniques applied in this case will carry on almost directly to the other models.

Using the explicit expression for the weight function $w^{\rCS}(x)$, we obtain
\be\label{eq:MLMRrCS}
 M^{-}(z) = z\,,
 \hspace{30pt}
 M^{+}(z) = rq^{\frac12}\,.
\ee
Therefore, setting $d_-=1$ in \eqref{cor:RecUni0} and substituting \eqref{eq:MLMRrCS} there, we obtain the recursion equation
\be
\label{eq:eigenvalueeqrCS}
 \hat{\sfA}^{\rCS}\cdot\sfZ^{\rCS}(\bp)=0\,,
\ee
with
\be
 \hat{\sfA}^{\rCS} := (1-x^{-}_0)
 - (rq^{\frac12}Qt^{-1}q)\left((1-t^{-1})(t/q)p_1-Q(t/q)x^{+}_{-1}\right) \,.
\ee
The recursion operator $\hat{\sfA}^{\rCS}$ has two types of terms: of degree zero and degree one. Hence, if we were to expand the partition function as a series in $p_k$, the equation would lead to a recursion of step one on the coefficients of the expansion. One can show, that this recursion has a unique solution, which can be found using combinatorial identities. Instead of doing so, we would like to solve the equation algebraically. Note, that $r$ enters the equation in a simple way, namely, we can reabsorb it by rescaling the variables as $p_k \mapsto r^{-k} p_k$, which corresponds to the action of the operator $r^{-\hD}$.

To solve the equation, let us define the ``\emph{gauge transformation}'' operator $\hG_{\rCS}$ as
\be\label{eq:gauge-rCS}
 \hG_{\rCS} := \Delta^{+}(Q)\sfT\exp\left( \sum_{k\geq1} (rq^{\frac12}Qt^{-1}q)^k
 \frac{p_k}{k(1-q^k)} \right)
 \sfT^{-1}\Delta^{+}(Q)^{-1}\,.
\ee
The main idea is to use this operator, to show that the operator in \eqref{eq:eigenvalueeqrCS} is conjugate to a much simpler one, whose kernel is one-dimensional and known explicitly. In particular, using relations from Lemma~\ref{lemma:T} and Corollary~\ref{lemma:Delta} we compute
\be
\ba
 & \hG_{\rCS}(1-x^-_0)\hG_{\rCS}^{-1} = \\
 &= \Delta^{+}(Q)\sfT\exp \left( \sum_{k\geq1}
 \frac{(rq^{\frac12}Qt^{-1}q)^k p_k}{k(1-q^k)} \right)
 (1-x^-_0)\exp \left( -\sum_{k\geq1} \frac{(rq^{\frac12}Qt^{-1}q)^k p_k}{k(1-q^k)} \right)
 \sfT^{-1}\Delta^{+}(Q)^{-1} \\
 &= (1-x^{-}_0)-(rq^{\frac12}Qt^{-1}q)\left((1-t^{-1})(t/q)p_1 - Q(t/q)x^{+}_{-1}\right)
\ea
\ee
which is exactly the recursion operator in \eqref{eq:eigenvalueeqrCS}.
Therefore we can prove the following.
\begin{lemma}
The equation \eqref{eq:eigenvalueeqrCS} for the refined Chern--Simons generating function has a unique solution given by
\be\label{eq:rCSGrep}
 \sfZ^{\rCS}(\bp)
 = \hG_{\rCS}\cdot1
 = \Delta^{+}(Q)\sfT\cdot\exp \left( \sum_{k\geq1} (rq^{\frac12}Qt^{-1}q)^k
 \frac{p_k}{k(1-q^k)} \right)\,.
\ee
\end{lemma}
\begin{proof}
The form of the solution follows directly from the calculation above. Taking into account the conjugation relation, we can write the recursion equation as 
\begin{equation}
\label{eq:eigen-rCS-vac}
 \hG_{\rCS}\,(1-x^-_0)\,\hG_{\rCS}^{-1} \cdot \sfZ^{\rCS}(\bp) = 0\,.
\end{equation}
The operator $\hG_{\rCS}$ is invertible hence it has trivial kernel. Therefore $\hG_{\rCS}^{-1} \cdot \sfZ^{\rCS}(\bp)$ lies in the kernel of $(1-x^-_0)$, which is one dimensional and consists of constant symmetric functions. This is because $1$ is the eigenvalue of $x^{-}_0$ when acting on the Macdonald polynomial $P_\vac\equiv1$.
Using the fact that the generating function is normalized such that $\sfZ^\rCS(0)=1$, we find that $\hG_{\rCS}^{-1} \cdot \sfZ^{\rCS}(\bp)=1$, which is equivalent to \eqref{eq:rCSGrep}.
\end{proof}
Superintegrability of the rCS model then follows as a simple corollary of this lemma. From \eqref{eq:rCSGrep}, we have
\begin{equation}
 \sfZ^{\rCS}(\bp)
 = \ev{\exp\left( \sum_{k\geq1} \frac{1-t^k}{1-q^k} \frac{p_k}{k}
 \sum_{i=1}^N x_i^k \right) }^{\rCS}
 = \Delta^{+}(Q)\sfT\cdot\exp\left( \sum_{k\geq1} (rq^{\frac12}Qt^{-1}q)^k
 \frac{p_k}{k(1-q^k)} \right)
\end{equation}
and by identifying the coefficients of $P_\lam(\bp)$ on both sides of the second equality, we obtain
\be
 \ev{P_\lam(\bx)}^{\rCS} = \Delta^{+}_\lam(Q)\, \sfT_\lam\,
 P_\lam\left(p_k=\frac{(rq^{\frac12}Qt^{-1}q)^k}{1-t^k}\right)\,.
\ee
This proves point \ref{it:SI:1} in Theorem~\ref{thm:SIqt}.

\begin{remark}
Let us now compare the superintegrability formula with the CMM identity \eqref{eq:CMM} for $\mu=\vac$.
Using the five-term relation \eqref{eq:5t} and the identity \eqref{eq:BHid}, we have
\be
\ba
 \sfZ^{\rCS}(\bp)
 &= \sfT \Delta^{+}(Q) \exp\left( \sum_{k\geq1} \frac{p_k}{k(1-q^k)} \right) \\
 &= \sfT \exp\left( \sum_{k\geq1} \frac{p_k}{k(1-q^k)} \right)
 \sfT \exp\left( \sum_{k\geq1} \frac{(-Q)^k p_k}{k(1-q^k)} \right) \\
 &= \sfT \exp\left( \sum_{k\geq1} \frac{1-Q^k}{1-q^k}\frac{p_k}{k} \right) \\
 &= \sum_{\lam} b_\lam \sfT_\lam\,P_\lam\left(p_k=\frac{1-Q^k}{1-t^k}\right) P_\lam\,.
\ea
\ee
This implies that the expectation value of $P_\lam$ can be written as
\be
 \ev{P_\lam}^{\rCS} = \sfT_\lam\,P_\lam\left(p_k=\frac{1-Q^k}{1-t^k}\right)
\ee
which matches precisely with \eqref{eq:CMM} upon the identification $u_\vac=\frac{1-Q}{1-t}$.
\end{remark}

\subsubsection{Proof of \ref{it:SI:2}}

Using that the weight function is the inverse of the one considered in the previous section (i.e.\ $\ell=-1$), we have
\be
 M^{-}(z) = q^\frac12\,,
 \hspace{30pt}
 M^{+}(z) = rz\,,
\ee
which leads to the recursion
\be
 \hat{\sfA}^{\overline{\rCS}} \cdot \sfZ^{\overline{\rCS}}(\bp)=0\,,
\ee
with
\be
 \hat{\sfA}^{\overline{\rCS}} := (1-x^{+}_0)+(q^\frac12 r^{-1} (qt^{-1}Q)^{-1})
 \left((t/q)(1-t^{-1})p_1+Q^{-1}q^{-1}x^{-}_{-1}\right)\,.
\ee
If we introduce, as before, a gauge transformation
\be\label{eq:gauge-op-rCS}
 \hG_{\overline{\rCS}} := \sfT^{-1}\Delta^{-}(Q^{-1})
 \exp\left( -\sum_{k\geq1} (q^\frac12 r^{-1} (qt^{-1}Q)^{-1})^k\frac{(t/q)^k p_k}
 {k(1-q^k)} \right)
 \Delta^{-}(Q^{-1})^{-1}\sfT\,,
\ee
we immediately find\footnote{We use the identity \eqref{eq:5t:e5}.}
\be
 \hat{\sfA}^{\overline{\rCS}} = \hG_{\overline{\rCS}}\,(1-x^{+}_0)\,
 \hG_{\overline{\rCS}}^{-1}
\ee
which leads to the unique solution
\be
\ba
 \sfZ^{\overline{\rCS}}(\bp) = \hG_{\overline{\rCS}}\,\cdot 1
 &= \sfT^{-1}\Delta^{-}(Q^{-1})
 \exp\left( -\sum_{k\geq1} (q^\frac12 r^{-1} (qt^{-1}Q)^{-1})^k\frac{(t/q)^kp_k}{k(1-q^k)} \right)\\
 &= \sfT^{-1}\Delta^{+}(Q)
 \exp\left( \sum_{k\geq1} (q^\frac12 r^{-1} (qt^{-1}Q)^{-2})^k\frac{p_k}{k(1-q^k)} \right)
\ea
\ee
from which it is straightforward to read the superintegrability formula for Macdonald functions. This proves \eqref{eq:SI:2} in Theorem~\ref{thm:SIqt}.

It is worth noting that in this case the recursion operator conjugates to the diagonal operator $(1-x^{+}_0)$ instead of $(1-x^{-}_0)$. This is in fact what one would expect by observing that $x^{-}_0$ and $x^{+}_0$ are related by the map that inverts the parameters $q$ and $t$.

\subsubsection{Comparison with skein-recursions}
Expectation values of Macdonald polynomials in the refined Chern--Simons model conjecturally compute some knot invariants, known as colored superpolynomials of the unknot \cite{Aganagic:2011sg,Aganagic:2012ne,Cherednik:2011nr}:
\begin{equation}
    \ev{P_\lam}^{\rCS} = \mathcal{P}_\lam^{\rm unknot}(q,t,Q)
\end{equation}
The generating function of rCS theory is then a refined version of the Ooguri--Vafa partition function for the unknot. In that context, it is useful to compare our equation \eqref{eq:eigenvalueeqrCS} to the skein recursion obtained in \cite{Ekholm:2019yqp,Ekholm:2020csl,Ekholm:2024ceb}.
Since the skein recursions are only known in the unrefined limit, we fix $t=q$ (and $r=q^{-\frac12}Q^{-1}$ for simplicity), which leads to
\begin{equation}
 \left\{(1-x^{+}_0)-(1-q) p_1+Q^{-1}x^{-}_{-1} \right\} \sfZ^{\overline{\rm CS}}(\bp)=0\,.
\end{equation}
Using notations of \cite{Ekholm:2024ceb}, we can identify the $\DIM$ generators with the skein algebra elements,
\begin{equation}
 1 \leftrightarrow \sfP_{0,0} \,,\quad
 x_{-k}^{-} \leftrightarrow \sfP_{-1,k}\,,\quad
 p_1 \leftrightarrow \sfP_{0,1}\,,\quad
 x_{-k}^{+}\leftrightarrow \sfP_{1,k}
\end{equation}
where $\sfP_{i,j}\in\mathrm{Sk}(T^2)$ are the generators of the skein algebra of the torus.
Taking into account different choices of normalization, our recursion operator corresponds to an operator in the skein:
\begin{equation}
    \sfA^{\rm skein}=\left( \sfP_{0,0} - \sfP_{1,0} + q^{\frac{1}{2}} \sfP_{0,1} + Q^{-1} q^{\frac{1}{2}} \sfP_{-1,1} \right)
\end{equation}
We want to compare this result to the skein recursion of \cite{Ekholm:2024ceb}. First, note that we use differently normalized $q$-variables, hence $q_{\rm{ELN}}=q^{1/2}$ and  $a_{\rm{ELN}} = Q^{1/2}$. Finally, to obtain the operator of \cite{Ekholm:2024ceb} (denoted as $\sfA_U^{(-1)}$ there) we rescale the operators as $\sfP_{i,j} \to \left( q^{\frac{1}{2}}Q^{-\frac{1}{2}} \right)^j \sfP_{i,j}$ and choose framing to be equal to $f=-1$.
Therefore we see that with these choices, we are able to match the equation for the generating function that we obtained at $t=q$ to the skein recursion relations of  \cite{Ekholm:2020csl,Ekholm:2024ceb}.

It is conjectured \cite{morton2017} that the refinement of the torus HOMFLY skein algebra is the elliptic Hall algebra or the quantum toroidal $\mathfrak{gl}_1$ algebra, hence we can also conjecture, that the equation obtained here for the refined CS model are the correct refinement of the skein recursion relations of \cite{Ekholm:2019yqp,Ekholm:2020csl,Ekholm:2024ceb}.

\subsection{Models with matter}

\subsubsection{Proof of \ref{it:SI:3}}

Let us start with $(\Nf,\Nbf)=(1,0)$ and $\ell=0$. We have
\be
 M^-(z) = 1-u_1 z\,,
 \hspace{30pt}
 M^+(z) = r
\ee
and hence we can write the recursion equation
\be\label{eq:Nf1recusion}
 \hat{\sfA}^{1\bar{0}}\cdot \sfZ^{1\bar{0}}(\bp) = 0
\ee
with
\begin{equation}
\label{eq:Nf1A1}
 \hat{\sfA}^{1\bar{0}} := (1-x^{-}_0)
 + u_1^{-1}\left( q^{-1} x^{-}_{-1}
 + Q(1-t^{-1})(t/q) p_1
 + r Q (1-t^{-1}) p_1
 - Q^2 r x^{+}_{-1})\right)
\end{equation}
To solve this equation, as we did before, we need to construct a gauge transformation operator which transforms the recursion equation into a simple eigenvalue equation for the trivial function $1$. Choosing
\be
\label{eq:gauge-transf-Nf1}
 \hG_{1\bar{0}} = \Delta^{+}(Q)\Delta^{+}(rqt^{-1}Q)
 \exp\left(\sum_{k\geq1}\frac{u_1^{-k}p_k}{k(1-q^k)}\right)
 \Delta^{+}(rqt^{-1}Q)^{-1}\Delta^{+}(Q)^{-1}
\ee
as gauge transformation, we find that $\hat{\sfA}^{1\bar{0}} = \hG_{1\bar{0}}(1-x^{-}_0)\hG_{1\bar{0}}^{-1}$, which leads to the solution
\be
 \sfZ^{1\bar{0}}(\bp) = \hG_{1\bar{0}}\cdot 1
 = \Delta^{+}(Q)\Delta^{+}(rqt^{-1}Q)
 \cdot \exp\left(\sum_{k\geq1}\frac{u_1^{-k}p_k}{k(1-q^k)}\right)\,.
\ee
This proves \eqref{eq:SI:3} in Theorem~\ref{thm:SIqt}.

\subsubsection{Proof of \ref{it:SI:4}}

This case corresponds to $(\Nf,\Nbf)=(0,1)$ and $\ell=0$. We have
\be
 M^{-}(z) = 1\,,
 \hspace{30pt}
 M^{+}(z) = r(1-v_1 z)
\ee
and hence we can write the recursion equation
\begin{equation}
\label{eq:Nf01recusion}
 \hat{\sfA}^{0\bar{1}}\cdot \sfZ^{0\bar{1}}(\bp) = 0
\end{equation}
with
\begin{equation}
\label{eq:Nf01A1}
 \hat{\sfA}^{0\bar{1}} :=  (1-x^{+}_0)
 +v_1^{-1}\left(
 (t/q) x^{+}_{-1}
 - (1-t^{-1})Q^{-1}(t/q)(1+(t/q) r^{-1}) p_1
 - Q^{-2}(t/q)r^{-1}q^{-1} x^{-}_{-1}
 \right)
\end{equation}
Choosing
\be
\label{eq:gauge-transf-Nf01}
 \hG_{0\bar{1}} = \sfT^{-2} \Delta^{+}(Q) \Delta^{+}(rqt^{-1}Q)
 \exp\left(\sum_{k\geq1}\frac{v_1^{-k}p_k}{k(1-q^k)}\right)
 \Delta^{+}(rqt^{-1}Q)^{-1} \Delta^{+}(Q)^{-1} \sfT^{2}
\ee
as gauge transformation, we find that $\hat{\sfA}^{0\bar{1}} = \hG_{0\bar{1}}(1-x^{+}_0)\hG_{0\bar{1}}^{-1}$, which leads to the solution
\be
 \sfZ^{0\bar{1}}(\bp) = \hG_{0\bar{1}}\cdot 1
 = \sfT^{-2} \Delta^{+}(Q) \Delta^{+}(rqt^{-1}Q)
 \cdot \exp\left(-\sum_{k\geq1}(r^{-1}(qt^{-1}Q)^{-2})^k
 \frac{v_1^{-k}p_k}{k(1-q^k)}\right)\,.
\ee
This proves \eqref{eq:SI:4} in Theorem~\ref{thm:SIqt}.

\subsubsection{Proof of \ref{it:SI:5}}

Similarly, for $(\Nf,\Nbf)=(2,0)$ and $\ell=0$, we have
\be
 M^{-}(z) = (1-u_1 z)(1-u_2 z)\,,
 \hspace{30pt}
 M^{+}(z) = 1
\ee
where we had to fix $r=1$ in order to have enough well-defined constraints to be able to solve the recursion uniquely.
In this case, the partition function satisfies the equation
\be\label{eq:D2}
 \hat{\sfA}^{2\bar{0}}\cdot \sfZ^{2\bar{0}}(\bp)=0
\ee
with
\begin{multline}
 \hat{\sfA}^{2\bar{0}} := (1-x^{-}_0)
 + (u_1^{-1}+u_2^{-1}) \left(q^{-1} x^-_{-1}+Q(1-t^{-1})(t/q)p_1 \right) \\
 - (u_1 u_2)^{-1}\left(q^{-2} x^-_{-2}-\frac{Q}{(1-q/t)}h_2
\left[(1-t^{-k})(1-(t/q)^k)p_k \right]-Q^2(t/q) x^{+}_{-2}\right)
\end{multline}
If we define the gauge transformation operator as
\be\label{eq:G20}
 \hat{G}_{2\bar{0}} := 
 \Delta^{+}(Q)
 \exp\left(\sum_{k\geq1}\frac{(u_1^{-k}+u_2^{-k})p_k}{k(1-q^k)}\right)
 \Delta^{+}(Q)^{-1}
\ee
then we find that $\hat{\sfA}^{2\bar{0}}=\hat{G}_{2\bar{0}}(1-x^{-}_0)\hat{G}_{2\bar{0}}^{-1}$ and the solution can be written as
\be
 \sfZ^{2\bar{0}}(\bp) = \hat{G}_{2\bar{0}}\cdot 1
 = \Delta^{+}(Q)\cdot
 \exp\left(\sum_{k\geq1}\frac{(u_1^{-k}+u_2^{-k})p_k}{k(1-q^k)}\right)\,.
\ee
This proves \eqref{eq:SI:5} in Theorem~\ref{thm:SIqt}.

\subsubsection{Proof of \ref{it:SI:6}}

For $(\Nf,\Nbf)=(0,2)$, $\ell=0$ and $r=1$, we have
\be
 M^{-}(z) = 1 \,,
 \hspace{30pt}
 M^{+}(z) = (1-v_1 z)(1-v_2 z)\,.
\ee
The partition function satisfies the equation
\begin{equation}
    \hat{\sfA}^{0\bar{2}}\cdot \sfZ^{0\bar{2}}(\bp)=0
\end{equation}
with
\begin{multline}
 \hat{\sfA}^{0\bar{2}} := (1-x^{+}_0)
 + (v_1v_2)^{-1}\Big(
 - (t/q)^2 x^{+}_{-2}
 - \frac{q^{-1}tQ^{-1}}{1-(q/t)} h_2((1-t^{-k})(1-(t/q)^k)p_k)
 + (t/q) Q^{-2} q^{-2} x^{-}_{-2} \\
 - (v_1+v_2)(t/q)((1-t^{-1})Q^{-1} p_1 - x^{+}_{-1})
 \Big)
\end{multline}
If we define the gauge transformation operator as
\be
 \hat{G}_{0\bar{2}} := 
 \sfT^{-1}\Delta^{+}(Q)
 \exp\left(-\sum_{k\geq1}(-qt^{-1}Q)^{-k}\frac{(v_1^{-k}+v_2^{-k})p_k}{k(1-q^k)}\right)
 \Delta^{+}(Q)^{-1}\sfT
\ee
then we find that $\hat{\sfA}^{0\bar{2}}=\hat{G}_{0\bar{2}}(1-x^{+}_0)\hat{G}_{0\bar{2}}^{-1}$ and the solution can be written as
\be
 \sfZ^{0\bar{2}}(\bp) = \hat{G}_{0\bar{2}}\cdot 1
 = \sfT^{-1}\Delta^{+}(Q)\cdot
 \exp\left(-\sum_{k\geq1}(-qt^{-1}Q)^{-k}\frac{(v_1^{-k}+v_2^{-k})p_k}{k(1-q^k)}\right)\,.
\ee
This proves \eqref{eq:SI:6} in Theorem~\ref{thm:SIqt}.

\subsubsection{Proof of \ref{it:SI:7}}

For $(\Nf,\Nbf)=(1,1)$ and $\ell=0$, we have
\be
 M^{-}(z) = 1-u_1z\,,
 \hspace{30pt}
 M^{+}(z) = r(1-v_1z)
\ee
from which we obtain the recursion equation
\begin{equation}
\label{eq:rec11}
 \sfA^{1\bar{1}}\cdot\sfZ^{1\bar{1}}(\bp) = 0
\end{equation}
with
\begin{multline}
\label{eq:Arec11}
 \sfA^{1\bar{1}} := \Big\{(1-x^{-}_0) + u_1^{-1}
 \left(q^{-1}x^{-}_{-1} + Q(1-t^{-1})(t/q) p_1\right)\Big\} \\
  - (t/q)\omega_1
 \Big\{(1-x^{+}_0)-v_1^{-1}\left(Q^{-1}(1-t^{-1})(t/q)p_1-(t/q)x^{+}_{-1}\right)\Big\}
\end{multline}
where, for convenience, we introduced the parameter
\be
 \omega_1 := r(qt^{-1}Q)^2\frac{v_1}{u_1}
\ee
The situation turns out to be rather different in this case as the degree-zero part of the recursion operator contains both $(1-x^{-}_0)$ and $(1-x^{+}_0)$. In this case we have not been able to find an operator that conjugates $\sfA^{1\bar{1}}$ to a simple eigenvalue equation. Nevertheless, if we define
\be
 \hG_{1\bar{1}} := \frac{\Delta^{+}(Q)\Delta^{+}(rqt^{-1}Q)}
 {\Delta^{+}(\omega_1)}
 \exp\left(\sum_{k\geq1}\frac{u_1^{-k}p_k}{k(1-q^k)}\right)
 \frac{\Delta^{+}(\omega_1)}
 {\Delta^{+}(Q)\Delta^{+}(rqt^{-1}Q)}
\ee
and
\be
 \hF_{1\bar{1}} := \frac{\Delta^{+}(Q)\Delta^{+}(rqt^{-1}Q)}
 {\Delta^{+}(\omega_1)}
 \sfT^{-1}
 \exp\left(-\sum_{k\geq1}\frac{(-1)^{k}u_1^{-k}p_k}{k(1-q^k)}\right)
 \sfT
 \frac{\Delta^{+}(\omega_1)}
 {\Delta^{+}(Q)\Delta^{+}(rqt^{-1}Q)}
\ee
we have the identity
\be
 \sfA^{1\bar{1}} = \hG_{1\bar{1}} (1-x^{-}_0) \hG_{1\bar{1}}^{-1}
 -(t/q)\,\omega_1\, \hF_{1\bar{1}} (1-x^{+}_0) \hF_{1\bar{1}}^{-1}
\ee
Now we can use \eqref{eq:BHid} to show that
\be
 \hG_{1\bar{1}}\cdot 1=\hF_{1\bar{1}}\cdot 1
\ee
which leads to the conclusion that
\be\label{eq:sol11}
 \sfZ^{1\bar{1}}(\bp) = \hG_{1\bar{1}} \cdot 1
 = \frac{\Delta^{+}(Q)\Delta^{+}(rqt^{-1}Q)}
 {\Delta^{+}(\omega_1)}
 \exp\left(\sum_{k\geq1}\frac{u_1^{-k}p_k}{k(1-q^k)}\right)
\ee
is a solution to the recursion \eqref{eq:rec11}. In fact $\sfA^{1\bar{1}}$ is a sum of two operators both of which annihilate the function \eqref{eq:sol11} separately. Because of this, we need to show that this is the unique solution (up to scalars). Using the fact that the spectrum of the operator $x^{\pm}_0$ is discrete, we can finally argue that, for generic enough values of the relative coefficient $\omega_1$, the kernel of $\sfA^{1\bar{1}}$ is equal to the kernel of $\hG_{1\bar{1}} (1-x^{-}_0) \hG_{1\bar{1}}^{-1}$ which is also equal to the kernel of $\hF_{1\bar{1}} (1-x^{+}_0) \hF_{1\bar{1}}^{-1}$.
This guarantees that \eqref{eq:sol11} is the unique solution of the recursion and concludes the proof of \eqref{eq:SI:6} in Theorem~\ref{thm:SIqt}.

\begin{remark}
Notice that if $v_1=0$ in \eqref{eq:Arec11}, the recursion operator $\hat{\sfA}^{1\bar{1}}$ reduces to the operator $\hat{\sfA}^{1\bar{0}}$ in \eqref{eq:Nf1A1}. The same limit holds at the level of the corresponding solutions.
On the other hand, if $u_1=0$, we have the limit $\lim_{u_1\to0}u_1\hat{\sfA}^{1\bar{1}}=-v_1Q(rqt^{-1}Q)\hat{\sfA}^{0\bar{1}}$.
This shows that \ref{it:SI:3} and \ref{it:SI:4} are corollaries of \ref{it:SI:7}.
\end{remark}

\subsubsection{Proof of \ref{it:SI:8}}

For $(\Nf,\Nbf)=(1,0)$ and $\ell=-1$, we have
\be
 M^{-}(z) = q^\frac12(1-u_1z) \,,
 \hspace{30pt}
 M^{+}(z) = rz
\ee
form which we obtain the recursion operator
\be
 \hat{\sfA}^{w} := (1-x^{+}_0)
 + q^\frac12 (rqt^{-1}Q)^{-1}
 \left( Q^{-1} q^{-1} x^{-}_{-1} + (t/q)(1-t^{-1})p_1 + u_1 Q^{-1}(1-x^{-}_0) \right)\,.
\ee
We observe now that this operator can be obtained from $\hat{\sfA}^{1\bar{1}}$ in \eqref{eq:Arec11} by rescaling the parameters as
\be
 r \mapsto \epsilon r\,,
 \hspace{30pt}
 v_1 \mapsto -\epsilon^{-1} q^{-\frac12}\,,
\ee
and then taking the limit $\epsilon\to0$. We easily deduce then that the solution takes the form
\be
\ba
 \sfZ^{w}(\bp)
 &= \lim_{\epsilon\to0}\frac{\Delta^{+}(Q)\Delta^{+}(\epsilon rqt^{-1}Q)}
 {\Delta^{+}(\epsilon r(qt^{-1}Q)^2\frac{(-\epsilon^{-1} q^{-\frac12})}{u_1})}
 \exp\left(\sum_{k\geq1}\frac{u_1^{-k}p_k}{k(1-q^k)}\right)\\
 &= \frac{\Delta^{+}(Q)}
 {\Delta^{+}(-q^{-\frac12}r(qt^{-1}Q)^2u_1^{-1})}
 \exp\left(\sum_{k\geq1}\frac{u_1^{-k}p_k}{k(1-q^k)}\right)\,.
\ea
\ee
This proves \eqref{eq:SI:8} in Theorem~\ref{thm:SIqt}.

\subsubsection{Proof of \ref{it:SI:9}}

For $(\Nf,\Nbf)=(0,1)$ and $\ell=1$, we have
\be
 M^{-}(z) = z \,,
 \hspace{30pt}
 M^{+}(z) = rq^\frac12(1-v_1z)
\ee
form which we obtain the recursion operator
\be
 \hat{\sfA}^{w} := (1-x^{-}_0)
 + (q^\frac12 r q t^{-1} Q )
 \left( Q (t/q) x^{+}_{-1} - (t/q)(1-t^{-1}) p_1 + v_1 Q (1-x^{+}_0 )\right)\,.
\ee
Just as in the previous case, we can obtain this operator from $\hat{\sfA}^{1\bar{1}}$ in \eqref{eq:Arec11} by rescaling the parameters as
\be
 r \mapsto \epsilon^{-1} r\,,
 \hspace{30pt}
 u_1 \mapsto -\epsilon^{-1} q^{-\frac12}\,,
\ee
and then taking the limit $\epsilon\to0$. We deduce that the solution takes the form
\be
\ba
 \sfZ^{w}(\bp)
 &= \lim_{\epsilon\to0}\frac{\Delta^{+}(Q)\Delta^{+}(\epsilon^{-1} rqt^{-1}Q)}
 {\Delta^{+}(\epsilon^{-1} r(qt^{-1}Q)^2\frac{v_1}{(-\epsilon^{-1} q^{-\frac12})})}
 \exp\left(\sum_{k\geq1}\frac{(-\epsilon q^{\frac12})^{k}p_k}{k(1-q^k)}\right) \\
 &= \lim_{\epsilon\to0}\frac{\Delta^{+}(Q)\Delta^{-}((\epsilon^{-1} rqt^{-1}Q)^{-1})\sfT}
 {\Delta^{+}(-q^{\frac12}r(qt^{-1}Q)^2 v_1)}
 (-\epsilon^{-1} rqt^{-1}Q)^{\hD}
 \exp\left(\sum_{k\geq1}\frac{(-\epsilon q^{\frac12})^{k}p_k}{k(1-q^k)}\right) \\
 &= \frac{\Delta^{+}(Q)\sfT}
 {\Delta^{+}(-q^{\frac12}r(qt^{-1}Q)^2 v_1)}
 \exp\left(\sum_{k\geq1}(q^{\frac12}rqt^{-1}Q)^k\frac{p_k}{k(1-q^k)}\right) \\
\ea
\ee
where we used \eqref{eq:Delta_id}.

This proves \eqref{eq:SI:9} in Theorem~\ref{thm:SIqt}.

\subsubsection{Proof of \ref{it:SI:10}}

For $(\Nf,\Nbf)=(2,2)$, $\ell=0$ and $r=1$, we have
\be
 M^{-}(z) = (1-u_1z)(1-u_2z)\,,
 \hspace{30pt}
 M^{+}(z) = (1-v_1z)(1-v_2z)
\ee
from which we obtain the recursion operator%
\footnote{Specializing the parameters as $u_2=\epsilon^{-1}$ and $v_2=r\epsilon^{-1}$, we can take the limit $\epsilon\to0$ and we obtain that the recursion operator $\hat{\sfA}^{2\bar{2}}$ reduces to the operator $\hat{\sfA}^{1\bar{1}}$.}
\be\label{eq:A22}
\ba
 \hat{\sfA}^{2\bar{2}}
 :=\,& \Big\{(1-x^{-}_0)
 +(u_1^{-1}+u_2^{-1})\left(-Qq^{-1}(1-t)p_1+q^{-1}x^{-}_{-1}\right) \\
 &+(u_1u_2)^{-1}\left(Qq^{-2}h_2\left((1-t^{k})p_k\right)
 -q^{-2} x^{-}_{-2}\right)
 \Big\} \\
 & - (t/q)\omega_2 \Big\{(1-x^{+}_0)
 +(v_1^{-1}+v_2^{-1})\left(-Q^{-1}(t/q)(1-t^{-1})p_1+(t/q)x^{+}_{-1}\right)\\
 &+(v_1v_2)^{-1}\left(Q^{-1}(t/q)^2h_2\left((1-t^{-k}) p_k\right)-(t/q)^2x^{+}_{-2}\right)
 \Big\}
\ea
\ee
where, for convenience, we introduced the parameter
\be
 \omega_2 := (qt^{-1}Q)^2\frac{v_1v_2}{u_1 u_2}\,.
\ee
We now want to find the kernel of this operator. First, we observe that we can rewrite as
\be\label{eq:A22b}
\ba
 \hat{\sfA}^{2\bar{2}}
 =\,& \frac{\Delta^{+}(Q)}{\Delta^{+}(\omega_2)}\Big\{(1-x^{-}_0)
 -q^{-1}(q^{-1}tQ^{-1}\omega_2)(v_1^{-1}+v_2^{-1})(1-t)p_1
 +q^{-1}(u_1^{-1}+u_2^{-1})x^{-}_{-1} \\
 &+q^{-2}(q^{-1}tQ^{-1}\omega_2)^2(v_1v_2)^{-1} h_2\left((1-t^{k})p_k\right)
 -q^{-2} (u_1u_2)^{-1}x^{-}_{-2}
 \Big\}\frac{\Delta^{+}(\omega_2)}{\Delta^{+}(Q)} \\
 & - (t/q)\omega_2 \frac{\Delta^{+}(Q)}{\Delta^{+}(\omega_2)}\Big\{(1-x^{+}_0)
 -(u_1^{-1}+u_2^{-1})(1-t^{-1})p_1
 +(v_1^{-1}+v_2^{-1})(q^{-1}tQ^{-1}\omega_2) x^{+}_{-1}\\
 &+(u_1u_2)^{-1} h_2\left((1-t^{-k}) p_k\right)
 -(q^{-1}tQ^{-1}\omega_2)^2(v_1v_2)^{-1}x^{+}_{-2}
 \Big\}\frac{\Delta^{+}(\omega_2)}{\Delta^{+}(Q)}\,.
\ea
\ee
Let us denote $E_k$ as a set of auxiliary powersum-like variables such that
\be
 E_k = E_k^{+}-E_k^{-}\,,
\ee
then we define a class of operators,
\be
\ba
 \hat{\sfB}^{-}_0[E_k] &= (1-x^{-}_0)
 + \sum_{l=1}^{\infty} (-q)^{-l} e_l\left((1-t^k)E_k^{-}\right)
 h_l\left( (1-t^k) p_k \right)
 - \sum_{l=1}^{\infty} (-q)^{-l} e_l\left((1-t^k)E_k^{+}\right) x^{-}_{-l} \\
 &= \oint\frac{\mathd z}{2\pi\mathi z}
 \exp\left(-\sum_{k\geq1}q^{-k}(1-t^k)\frac{z^{-k}E^{-}_k}{k}\right)
 \exp\left(\sum_{k\geq1}(1-t^k)\frac{z^k p_k}{k}\right) \\
 &- \oint\frac{\mathd z}{2\pi\mathi z}
 \exp\left(-\sum_{k\geq1}q^{-k}(1-t^k)\frac{z^{-k}E^{+}_k}{k}\right)
 x^{-}(z)
\ea
\ee
and
\be
\ba
 \hat{\sfB}^{+}_0[E_k] &= (1-x^{+}_0)
 + \sum_{l=1}^{\infty} (-1)^l e_l\left((1-t^k)E^{+}_k\right)
 h_l\left((1-t^{-k}) p_k \right)
 - \sum_{l=1}^{\infty} (-1)^l e_l\left((1-t^k)E^{-}_k\right) x^{+}_{-l} \\
 &= \oint\frac{\mathd z}{2\pi\mathi z}
 \exp\left(-\sum_{k\geq1}(1-t^k)\frac{z^{-k}E^{+}_k}{k}\right)
 \exp\left(\sum_{k\geq1}(1-t^{-k})\frac{z^k p_k}{k}\right) \\
 &- \oint\frac{\mathd z}{2\pi\mathi z}
 \exp\left(-\sum_{k\geq1}(1-t^{k})\frac{z^{-k}E^{-}_k}{k}\right)
 x^{+}(z)
\ea
\ee
as the zero-modes of certain currents. Comparing with the expression in \eqref{eq:A22b}, we now notice that we can write
\begin{equation}
 \hat{\sfA}^{2\bar{2}} = \frac{\Delta^{+}(Q)}{\Delta^{+}(\omega_2)}
 \left(\hat{\sfB}^{-}_0
 [E_k]-(t/q)\omega_2\,\hat{\sfB}^{+}_0
 [E_k] \right)
 \frac{\Delta^{+}(\omega_2)}{\Delta^{+}(Q)}
\end{equation}
for an appropriate choice of $E_k$, namely
\be
 E^{+}_k = \frac{u_1^{-k}+u_2^{-k}}{1-t^k}\,,
 \hspace{30pt}
 E^{-}_k = (q^{-1}tQ^{-1}\omega_2)^k\,\frac{v_1^{-k}+v_2^{-k}}{1-t^k}\,.
\ee
If we can now prove that $\hat{\sfB}^\pm_0$ share the same kernel and that this kernel is 1-dimensional, then, assuming $\omega_2$ generic, it follows that $\hat{\sfA}^{2\bar{2}}$ must also have a 1-dimensional kernel, and the solution to the recursion equation is unique.
In order to show this, we observe that conjugating $\hat{\sfB}^\pm_0$ by an appropriate plethystic exponential function as follows
\be
 \tilde{\sfB}^{\pm}_0[E_k] :=
 \exp\left(-\sum_{k\geq1} \frac{1-t^k}{1-q^k} \frac{E_k p_k}{k}\right)
 \cdot \hat{\sfB}^{\pm}_0[E_k] \cdot
 \exp\left(\sum_{k\geq1} \frac{1-t^k}{1-q^k} \frac{E_k p_k}{k}\right)
\ee
produces new operators
\be
\ba
 \tilde{\sfB}^{+}_0[E_k] &= \oint\frac{\mathd z}{2\pi\mathi z}
 \exp\left(-\sum_{k\geq1} (1-t^k)\frac{z^{-k} E^{+}_k}{k}\right)
 \left\{
 \exp\left(\sum_{k\geq1} (1-t^{-k})\frac{z^k p_k}{k}\right)
 - x^{+}(z) \right\} \\
 \tilde{\sfB}^{-}_0[E_k] &= \oint\frac{\mathd z}{2\pi\mathi z}
 \exp\left(-\sum_{k\geq1} q^{-k}(1-t^k)\frac{z^{-k} E^{-}_k}{k}\right)
 \left\{
 \exp\left(\sum_{k\geq1} (1-t^{k})\frac{z^k p_k}{k}\right)
 - x^{-}(z) \right\}
\ea
\ee
and it is now straightforward to show that\footnote{Notice that this is indeed true for any values of $E^{\pm}_k$.}
\be
 \tilde{\sfB}^{\pm}_0[E_k] \cdot 1 = 0
\ee
where the constant symmetric function $1\in\Lambda$ is the unique (up to scalar) solution to this equation.
In turn, this implies that $\hat{\sfB}^{\pm}_0[E_k]$ both have the same 1-dimensional kernel spanned by the function
\be
 \exp\left(\sum_{k\geq1} \frac{1-t^k}{1-q^k} \frac{E_k p_k}{k}\right)\,.
\ee
Finally, we find that the solution to $\hat{\sfA}^{2\bar{2}}\cdot\sfZ^{2\bar{2}}(\bp)=0$ is given by
\be
 \sfZ^{2\bar2}(\bp) = \frac{\Delta^{+}(Q)}{\Delta^{+}(\omega_2)}
 \exp\left(\sum_{k\geq1} \left((u_1^{-k}+u_2^{-k})
 - (v_1^{-k}+v_2^{-k}) \left( q^{-1} t Q^{-1} \omega_2 \right)^k\right)
 \frac{p_k}{k(1-q^k)}\right)\,.
\ee
We deduce that the specialization homomorphism $\varphi^{2\bar2}$ can be defined as
\be
 \varphi^{2\bar2}(p_k) = E_k\,.
\ee
This proves \eqref{eq:SI:10} in Theorem~\ref{thm:SIqt}.

While we have not been able to conjugate $\hat{\sfA}^{2\bar2}$ to a simple operator in this case, we can still formally define $\hG_{2\bar2}$ as
\be
 \hG_{2\bar2} := \frac{\Delta^{+}(Q)}{\Delta^{+}(\omega_2)}
 \exp\left(\sum_{k\geq1} \left((u_1^{-k}+u_2^{-k})
 - (v_1^{-k}+v_2^{-k}) \left( q^{-1} t Q^{-1} \omega_2 \right)^k\right)
 \frac{p_k}{k(1-q^k)}\right)
 \frac{\Delta^{+}(\omega_2)}{\Delta^{+}(Q)}\,,
\ee
then we have that $\sfZ^{2\bar2}(\bp)= \hG_{2\bar2}\cdot 1$.

\subsection{Relation to the \texorpdfstring{$W$}{W}-representation}

Formulas \eqref{eq:SI:1}--\eqref{eq:SI:10} can be written in a $W$/cut-and-join representation form. Namely, we can write
\be
\ba
 \sfZ^w(\bp)
 &= \hG_w\cdot 1
 = \hat{\sfC}^w \cdot \exp \left( \sum_{k\geq1} \frac{1-t^k}{1-q^k}
 \frac{\varphi_k^w p_k}{k}\right)
 \cdot(\hat{\sfC}^w)^{-1} \cdot 1 \\
 &= \exp \left( \sum_{k\geq1} \frac{1-t^k}{1-q^k} \frac{\varphi^w_k}{k}\,
 \hat{\sfC}^w\, p_k\, (\hat{\sfC}^w)^{-1}\right)\cdot 1
 = \exp \left( \sum_{k\geq1} \frac{1-t^k}{1-q^k} \frac{\varphi_k^w}{k}\, \hW^{w}_{-k}\right)
 \cdot 1
\ea
\ee
where the operators $\hW^{w}_{-k}:=\hat{\sfC}^w\cdot p_k\cdot(\hat{\sfC}^w)^{-1}$ are $W$-like operators.
For given models, these coincide with the operators constructed in \cite{Liu:2023trf} by naively reverse-engineering the superintegrabilty formulas (the operators $\hat{\mathbf{O}}$ in \cite{Liu:2023trf} are called $\hat{\sfC}$ here).

\section{Orthogonal polynomials and superintegrability}
\label{sec:orthogonal}

\subsection{Interpolation Macdonald polynomials and orthogonal polynomials for rCS}

As observed in \eqref{eq:CMM}, the Cherednik--Macdonald--Mehta identity gives a generalization of Macdonald superintegrability in the case of the rCS matrix model, which allows to compute explicitly the ensemble average of products of two Macdonald polynomials.
One may regard this as the definition of an Hermitian inner product on the space of symmetric polynomials in the variables $\bx$. This can be extended to an inner product on the space of symmetric functions as
\be\label{eq:rCSpairing}
 (P_\lam,P_\mu)^{\rCS} := \ev{ P_\lam \, P_\mu }^{\rCS}
 = \sfT_\lam\, P_\lam(u_\vac)\, P_\mu(u_\lam)\, \sfT_\mu\,,
\ee
with $u_\lam$ as in \eqref{eq:symbols3}.
It is not hard to see that this product is related to the so called Fourier/Hopf pairing \cite{Cherednik:1995mac,Okounkov:2001are,Beliakova:2021cyc}, $(\,,\,)_{\rm F}$, by a twist by the framing operator, i.e.\
\be\label{eq:rCSpairingFourier}
 (P_\lam,P_\mu)^{\rCS} = (\sfT\cdot P_\lam,\sfT\cdot P_\mu)_{\rm F}\,,
\ee
where the Fourier/Hopf pairing has the property
\be
 (f,P_\lam)_{\rm F} = f(u_\lam)\, P_\lam(u_\vac)
\ee
for any symmetric function $f\in\Lambda$. An orthogonal basis for the Fourier/Hopf pairing is given by the interpolation Macdonald polynomials $P^\ast_\lam$ studied in \cite{Knop:1996sym,Sahi:1996int,Sahi:1996dif,Okounkov:1997shi}. These are equivalently defined by their interpolation property
\be
 P^\ast_\mu(u_\lam) = 0
 \hspace{30pt}
 \text{if}
 \hspace{30pt}
 \mu\not\subseteq\lam\,,
\ee
and it was shown in \cite{Garsia2001,Carlsson:2013jka} that they can be obtained from the usual Macdonald functions via an exponential operator as
\be
 P^\ast_\lam = \Delta^{+}(Q)^{-1}
 \exp\left(-\sum_{k\geq1}\frac{1}{1-t^k}\frac{\partial}{\partial p_k}\right)
 \Delta^{+}(Q)\cdot P_\lam\,,
\ee
in a similar way to how the classical orthogonal polynomials can be obtained from Schur/Jack via some exponential operators in the spirit of the Lassalle--Nekrasov correspondence \cite{Lassalle:1991po,Nekrasov:1997jf}. From \eqref{eq:rCSpairingFourier} it is then clear that a basis of orthogonal functions for the rCS pairing is given by the functions
\be\label{eq:LN-rCS}
 \sfW^{\rCS}_\lam := \sfT_\lam\, \sfT^{-1}\cdot P^\ast_\lam
 = \sfT^{-1}\Delta^{+}(Q)^{-1}
 \exp\left(-\sum_{k\geq1}\frac{1}{1-t^k}\frac{\partial}{\partial p_k}\right)
 \Delta^{+}(Q)\sfT\cdot P_\lam
\ee
where the additional factor of $\sfT_\lam$ is purely conventional and it has been added for normalization purposes.
Upon close inspection of the operator in \eqref{eq:LN-rCS}, we realize that this is the adjoint inverse of the gauge transformation $\hG_{\rCS}$ in \eqref{eq:gauge-rCS}, so that\footnote{In order to have an exact match with $\hG_{\rCS}$, we have to fix $r=q^{-\frac12}Q^{-1}tq^{-1}$. Alternatively, we could rescale the power sums $p_k$ in the interpolation functions by the inverse of the same factor.}
\be
 \sfW^{\rCS}_\lam = (\hG_{\rCS}^\perp)^{-1}\cdot P_\lam
 = P_\lam + (\text{lower degree terms})\,.
\ee
Following the work of Olshanski in \cite{Olshanski:2019int}, one can define a dual basis w.r.t.\ the ordinary Macdonald inner product which we denote as $\sfZ^{\rCS}_\mu$ and that satisfies the orthogonality condition
\be
 \langle \sfW^{\rCS}_\lam,\sfZ^{\rCS}_\mu\rangle_{q,t} = b_\lam^{-1} \delta_{\lam,\mu}\,.
\ee
which is equivalent to the Cauchy identity
\be
 \exp\left( \sum_{k\geq1}\frac{1-t^k}{1-q^k}\frac{p_k(\bx)p_k(\by)}{k}\right)
 = \sum_\lam b_\lam\,\sfW^{\rCS}_\lam(\bx)\,\sfZ^{\rCS}_\lam(\by)\,.
\ee
By construction, we find that the $\sfZ^{\rCS}_\mu$ also admit and exponential operator construction given by the formula
\be
 \sfZ^{\rCS}_\mu = \hG_{\rCS}\cdot P_\mu
 = P_\mu + (\text{higher degree terms})\,,
\ee
and now it becomes obvious that the generating function of the rCS matrix model is just a special case corresponding to $\mu=\vac$, as evident from \eqref{eq:rCSGrep}.
Correspondingly, \eqref{eq:eigen-rCS-vac} is a special case of the more general eigenvalue equation
\be
 \hG_{\rCS}\,(x^\vee_\mu-x^{-}_0)\,\hG_{\rCS}^{-1} \cdot \sfZ^{\rCS}_\mu = 0\,.
\ee
with eigenvalue $x^\vee_\mu$ as in \eqref{eq:symbols2}.

In this specific case, we can even relate directly the rCS inner product to the Macdonald inner product by the formula
\be\label{eq:rCS-pairing}
 (f,g)^{\rCS} = \langle f, \hG_{\rCS}\cdot\hK^{\rCS}\cdot\hG_{\rCS}^\perp\cdot g \rangle_{q,t}
\ee
where $\hK^{\rCS}:=\hat{\sfC}^{\rCS}\sfT^2(-Qt^{-1})^{\hD}$, which follows from \cite[Theorem I.2]{Garsia2001}.

For the case of rCS with opposite orientation, we similarly have that $\sfW^{\overline{\rCS}}_\lam:=(\hG_{\overline{\rCS}}^\perp)^{-1}\cdot P_\lam$ form an orthogonal basis w.r.t.\ the inner product induced by the matrix model integral. Moreover, one can define analogous interpolation Macdonald polynomials\footnote{Again, we have chosen $r=q^{-\frac12}Q^{-1}tq^{-1}$ in \eqref{eq:gauge-op-rCS} for simplicity.}
\be
 P^{\overline{\ast}}_\lam := \sfT_\lam\,\sfT^{-1}\cdot
 \sfW^{\overline{\rCS}}_\lam
\ee
which statisfy
\be
 P^{\overline{\ast}}_\mu(u^\vee_\lam) = 0
 \hspace{30pt}
 \text{if}
 \hspace{30pt}
 \mu\not\subseteq\lam\,.
\ee

\subsection{Orthogonal polynomials for models with matter}

Given our previous discussion of the orthogonal polynomials for the rCS matrix models, we would like to generalize those ideas to other superintegrable models as well. For this reason, whenever the matrix model with weight function $w(x)$ is Macdonald superintegrable as in Definition~\ref{def:Macdonald-SI}, we are lead to define the following two families of symmetric functions,
\be\label{eq:def-WZ}
 \sfW^w_\lam := (\hG_w^\perp)^{-1}\cdot P_\lam\,
 \hspace{30pt}
 \sfZ^w_\lam := \hG_w\cdot P_\lam\,,
\ee
for $\hG_w$ defined as
\be
 \hG_w := \hat{\sfC}^w\cdot\exp\left(\sum_{k\geq1}\frac{1-t^k}{1-q^k}
 \frac{\varphi^w_k\,p_k}{k}\right)\cdot(\hat{\sfC}^w)^{-1}\,,
\ee
which then implies
\be
 (\hG_w^{-1})^\perp = (\hat{\sfC}^w)^{-1}\cdot
 \exp\left(-\sum_{k\geq1} \varphi^w_k \frac{\partial}{\partial p_k}\right)
 \cdot \hat{\sfC}^w\,.
\ee
The functions $\sfW^w_\lam$ are non-homogeneous polynomials of top degree $|\lam|$, while $\sfZ^w_\lam$ are formal series of symmetric functions, whose lowest degree term has degree $|\lam|$.

Following from the definition, we have the orthogonality relation and Cauchy identity
\be
 \langle\sfW^w_\mu,\sfZ^w_\lam\rangle_{q,t} = b_\lam^{-1}\delta_{\lam,\mu}\,,
 \hspace{30pt}
 \exp\left(\sum_{k\geq1}\frac{1-t^k}{1-q^k}\frac{p_k(\bx)\,p_k(\by)}{k}\right)
 = \sum_\lam b_\lam\,\sfW^w_\lam(\bx)\,\sfZ^w_\lam(\by)\,,
\ee
as well as explicit formulas for the expansion in the Macdonald basis,
\be
\ba
 \sfW^w_\mu &= \sum_{\lam\subset\mu} \dfrac{\sfC^w_\mu}{\sfC^w_\lam}
 P_{\mu/\lam} \left(p_k=-\varphi^w_k \right) P_\lam\,, \\
 \sfZ^w_\mu &= \sum_{\mu\subset\lam} \dfrac{\sfC^w_\lam}{\sfC^w_\mu}
 P_{\lam/\mu} \left(p_k=\varphi^w_k \right) P_\lam\,,
\ea
\ee
and the eigenvalue equations
\be\label{eq:eigenop-WZ+}
 (\hG_w^\perp)^{-1} \, x_0^{+} \hG_w^\perp \cdot
 \sfW^w_\mu = x_\mu \, \sfW^w_\mu\,,
 \hspace{30pt}
 \hG_w \, x_0^{+} \hG_w^{-1} \cdot
 \sfZ^w_\mu = x_\mu \, \sfZ^w_\mu\,,
\ee
\be\label{eq:eigenop-WZ-}
 (\hG_w^\perp)^{-1} \, x_0^{-} \hG_w^\perp \cdot
 \sfW^w_\mu = x_\mu^\vee \, \sfW^w_\mu\,,
 \hspace{30pt}
 \hG_w \, x_0^{-} \hG_w^{-1} \cdot
 \sfZ^w_\mu = x_\mu^\vee \, \sfZ^w_\mu\,.
\ee

At this point we would like to claim that the $\sfW^w_\lam(\bx)$ are the orthogonal polynomials for the matrix model with weight function $w(x)$, however, computer experiments suggest that this is not true in general. In particular, we have checked that this seems to be the case only when the recursion operator $\hat{\sfA}^w$ satisfies either
\be
 \hat{\sfA}^w = \hG_w(1-x^{-}_0)\hG_w^{-1}\,,
\ee
or
\be
 \hat{\sfA}^w = \hG_w(1-x^{+}_0)\hG_w^{-1}\,,
\ee
or perhaps a linear combination of the two.
When this is the case (e.g.\ in the models from \ref{it:SI:1} to \ref{it:SI:6}), we conjecture that the matrix model integral induces an inner product such that
\be
 \ev{ \sfW_\lam(\bx)\,\sfW_\mu(\bx) }^w
 \stackrel{?}{=} b_\lam^{-1} K^w_\lam \delta_{\lam,\mu}
\ee
for some model-dependent constant $K^w_\lam$. We can then formally define the symmetric bilinear pairing
\be
 (f,g)^w := \langle f, \hG_{w}\cdot\hK^{w}\cdot\hG_{w}^\perp\cdot g \rangle_{q,t}
 \stackrel{?}{=} \ev{f(\bx)\,g(\bx)}^w
\ee
in analogy with the case of rCS in \eqref{eq:rCS-pairing}, where $\hK^{w}$ is diagonal on Macdonald functions with eigenvalue $K^{w}_\lam$. The functions $\sfW^w_\lam$ are, by construction, mutually orthogonal w.r.t.\ this inner product.

We observe that, as a straightforward corollary, we have the following exact expression (which seems to be related to a property called strong superintegrability in \cite{Mironov:2022gvd})
\be
\ba
 \ev{\sfW^w_\mu(\bx)\,P_\lam(\bx)}^{w}
 &= b^{-1}_\mu K^w_\mu \frac{\sfC^w_\lam}{\sfC^w_\mu} P_{\lam/\mu}(p_k=\varphi_k^w)\\
 &= b^{-1}_\mu \frac{K^w_\mu}{\sfC^w_\mu} \frac{P_{\lam/\mu}(p_k=\varphi_k^w)}
 {P_\lam(p_k=\varphi_k^w)} \ev{ P_\lam(\bx) }^w \,.
\ea
\ee

By the reproducing property of the Macdonald kernel, we also obtain
\be
 \dev{ f(\bx) }^{w}
 = \ev{ f(\bx)\, \exp\left( \sum_{k\geq1}\frac{1-t^k}{1-q^k}
 \frac{p_k}{k} \sum_{i=1}^N x_i^k \right) }^w
 = \hG_w \hK^w \hG^\perp_w \cdot f 
\ee
and in the case $f=\sfW^w_\lam$, we find
\be\label{eq:ZstarrcS}
 \dev{ \sfW^w_\lam(\bx) }^{w} = K^w_\lam\, \sfZ^w_\lam(\bp)\,.
\ee
For $\lam=\vac$, $K^w_\vac=1$ and one recovers the usual matrix model generating function \eqref{eq:gen-func-qt}.

In general, we expect that one should be able to prove the conjectural identity between $\sfW^w_\lam$ and the orthogonal polynomials by showing that the eigenoperators $(\hG_w^\perp)^{-1} \, x_0^{\pm} \hG_w^\perp$ are Hermitian w.r.t.\ the weight $w(x)$, i.e.\ they are self-adjoint w.r.t.\ the matrix model inner product (in finitely many variables).
In the case of the $(\Nf,\Nbf)=(2,0)$ model with $r=1$ and $\ell=0$, this was done in \cite{Baker:1997mul} and, as we show explicitly in Appendix~\ref{app:AlSalamCarlitz}, the functions $\sfW_\lam^{2\bar0}(\bx)$ can be identified with the Al-Salam--Carlitz polynomials. We leave the proof of the remaining cases for future work.

We conclude with a remark on the cases where the conjugation of the recursion operator by $\hG_w$ does not lead to a Macdonald operator. Let us consider the $(\Nf,\Nbf)=(2,2)$ model with $r=1$ and $\ell=0$ for concreteness (the other cases can be obtained through various limits). In this case, the orthogonal polynomials are know to be the multivariable big $q$-Jacobi polynomials $P^B_\lam(\bx)$ (see \cite{Stokman1997:mul}).
While it is not true that the operator $(\hG_{2\bar2}^\perp)^{-1}$ generates big $q$-Jacobi symmetric functions out of ordinary Macdonald functions $P_\lam$, we do have that the orthogonal functions $P^B_\lam$ satisfy the eigenvalue equation
\be
 (\hat{\sfA}^{2\bar2})^\perp \cdot P^B_\lam
 = \left((1-x^\vee_\lam)-(qt^{-1}Q^2\frac{v_1v_2}{u_1u_2})(1-x_\lam)\right)P^B_\lam\,,
\ee
with $(\hat{\sfA}^{2\bar2})^\perp$ the adjoint of the recursion operator in \eqref{eq:A22} and $x_\lam$ as in \eqref{eq:symbols2}.
This is the stable (large $N$) limit of the eigenvalue equation in \cite[Theorem~5.7]{Stokman1997:mul}.
This implies that the generating function $\sfZ^{2\bar2}(\bp)=\hG_{2\bar2}\cdot1$ is the orthogonal dual to $P^B_\vac(\bx)$ w.r.t.\ the Macdonald product, since it does satisfy the dual eigenvalue equation for $\lam=\vac$, namely $\hat{\sfA}^{2\bar2}\cdot\sfZ^{2\bar2}(\bp)=0$.

We expect that there should exist a different gauge transformation $\hG_B\neq\hG_{2\bar2}$ in this case, such that $(\hG_B^\perp)^{-1}\cdot P_\lam = P^B_\lam$ which however satisfies $\hG_B\cdot 1=\hG_{2\bar2}\cdot 1=\sfZ^{2\bar2}(\bp)$, thanks to some non-trivial identity. By the general logic of the Lassalle--Nekrasov correspondence, one then has
\be
 \hat{\sfA}^{2\bar2} =
 \hG_B\left((1-x^{-}_0)-(qt^{-1}Q^2\frac{v_1v_2}{u_1u_2})(1-x^{+}_0)\right)\hG_B^{-1}\,.
\ee
We are not aware of any explicit expression for the operator $\hG_B$ as an element of $\DIM$ in the existing literature. We should remark however, that some formulas for the matrix elements of $(\hG_B^\perp)^{-1}$ in the Macdonald basis have been derived in \cite[Corollary~3.3]{Olshanski:2020mac} via BC-type interpolation polynomials. It would be interesting to reformulate these results in terms of the action of the quantum toroidal algebra of $\mathfrak{gl}_1$.

\appendix
\addtocontents{toc}{\protect\setcounter{tocdepth}{1}}

\section{Jackson integrals and contour integrals}\label{sec:AppendixContours}
The matrix models discussed in the paper admit two equivalent representations: one as a contour integral and the other as a Jackson integral. Here we show that the latter can be transformed into the contour integral form by inserting an appropriate $q$-constant function.

Define the function $\gamma_q(x)$ as
\begin{equation}
 \gamma_q(x) := \exp\left(-\dfrac{\log^2(-x)}{2\log q} \right) \frac{x^{1/2}}
 {\theta_q(x)}
\end{equation}
where
\begin{equation}
 \theta_q(x) := (q;q)_{\infty}(x;q)_\infty(q x^{-1};q)_\infty
\end{equation}
is the Jacobi theta function.
The function $\gamma_q(x)$ is a $q$-constant, i.e.\
\begin{equation}
    \gamma_q(q x) = \gamma_q(x) \, \quad \Rightarrow \quad D^q_x \gamma_q(x) = 0 
\end{equation}
Then, the following equality allows us to transform a contour integral into a Jackson integral:
\begin{equation}
\label{eq:fromcontourtoJackson}
 \int_{\mathcal{C}_{q,a}}\mathd x\, \gamma_q(q x a^{-1})f(x)
 = - \dfrac{\exp{\left(\dfrac{\pi^2}{2 \log q}\right)} }{(q,q)_\infty^3}
 \sum_{n=0}^{\infty} (q^{n}a) f(q^{n}a) = \kappa_q \int_0^a \mathd_q x f(x)  
\end{equation}
where $\mathcal{C}_{q,a}$ is the contour that goes around the poles at $x=q^{n}a$, for $n\geq0$, coming from zeros of the theta function and avoids any poles of $f(x)$, if any.

Therefore, we define the following two functions:
\be
\label{eq:def-gamma_q}
\ba
 \gamma_q(x|a) &= \frac{1}{\kappa_q} \gamma_q(q x a^{-1})\,,\\
 \gamma_q(x|a,b) &= \frac{1}{\kappa_q}\Big(\gamma_q(q x b^{-1})-\gamma_q(q x a^{-1})\Big)\,.
\ea
\ee
Both functions are $q$-constant, since they involve only an argument rescaling compared to $\gamma_q(x)$. With these definitions, we have that
\begin{equation}
\ba
 & \int_{\mathcal{C}_{q,a}} \mathd x\,\gamma_q(x|a) f(x) =  \int_0^a \mathd_q x f(x)  \\
 & \int_{\mathcal{C}_{q,a,b}}\mathd x\, \gamma_q(x|a,b) f(x) = \int_a^b \mathd_q x f(x)  
\ea
\end{equation}
Here $\mathcal{C}_{q,a,b}$ goes around the respective poles of both $\gamma_q(q x a^{-1})$ and $\gamma_q(q x b^{-1})$.

\section{Semiclassical limit}
In this section we consider the semiclassical limit, which is defined by setting
\be
 q = \mathe^{\hbar} \,,
 \hspace{30pt}
 t = \mathe^{\hbar\beta}
\ee
and then taking the limit
\begin{equation}
    \hbar \rightarrow 0
\end{equation}
This limit is well understood at the level of superintegrability formulas for most of the examples that we study in this paper. After reviewing some of these results, we will explain how the standard cut-and-join/$W$-operator equations are reproduced from our recursions equations, in this limit.
At the level of the algebra, we expect that elements of the quantum toroidal algebra degenerate to elements of the affine Yangian $Y(\hat{\mathfrak{gl}_1})$ \cite{tsymbaliuk2017affine}.

Throughout this section we will use the $\Nf=1$ model as our main example, however most of the formulas here generalize to other models as well (with the exception of rCS).

First, it is important to recall that the semiclassical limit can be taken at the level of the integrand of the matrix model. We will sketch how it works in some examples. There are many caveats, as, for example, the possibility to take the limits in different ways to obtain new matrix models \cite{Mishnyakov:2024cgl}, which we will not discuss here.

Consider, for example, the model with $(\Nf,\Nbf)=(1,0)$ and $\ell=0$ in its Jackson integral formulation. To take the limit, we fix $u_1 = -(1-q^{-1})a_1$ and $r=q^\alpha$. Then, the potential becomes
\be
 \lim_{q\to1} x^{\frac{\log r}{\log q}} (-q(1-q^{-1})a_1x;q)_\infty
 = x^\alpha \exp(-a_1 x) \,.
\ee
At the same time the Jackson integral becomes an ordinary integral, and taking into account, that the upper limit of integration goes to $\infty$ under the limit, we obtain
\begin{multline}
 \lim_{q\to1} \int_0^{(qu_1)^{-1}} \prod_{i=1}^N \mathd_q x_i\,
 \Delta_{q,q^{\beta}}(\bx) f(\bx) \prod_{i=1}^N x_i^{\beta(N-1)+\alpha}
 (qu_1x_i;q)_{\infty} = \\
 = \int_0^\infty \prod_{i=1}^N \mathd x_i\,
 \prod_{i\neq j}(x_i-x_j)^\beta\, f(\bx) \prod_{i=1}^{N}x_i^{\alpha} \mathe^{a_1x_i} 
\end{multline}
which is nothing but the Wishart--Laguerre $\beta$-ensemble.

Clearly we would like the associated algebraic structures to also have a limit. The limit of the $q,t$-Virasoro constraints has been studied in \cite{Cassia:2020uxy}.
Below we show that the limit of the equation \eqref{eq:Nf1recusion} reproduces the cut-and-join/$W$-equation for the WL $\beta$-ensemble.

We need to compute the $\hbar$-expansion of the operator $\hat{\sfA}_{1\bar{0}}$ in \eqref{eq:Nf1A1}.
In the level-1 Fock representation of $\DIM$, this can be done directly by expanding the vertex operators $x^\pm(z)$ at leading order in $\hbar$.
\begin{lemma}
The semiclassical limit of the recursion operator in \eqref{eq:Nf1A1} reproduces the cut-and-join recursion operator of the Wishart--Laguerre $\beta$-ensemble, in fact, we have
\begin{equation}
\label{eq:classical-limit-D1}
 \hat{\sfA}_{1\bar{0}} = -\hbar^2\beta\left( \hD - \frac1{a_1}\hW^{\beta}_{-1}\right) + O(\hbar^3)\,,
\end{equation}
where $\hD$ is the degree operator in \eqref{eq:degree-op} and
\begin{multline}
 \hW^{\beta}_{-1}
 = \sum_{a,b=1}^\infty\left(
 \,a\, b\, p_{a+b+1} \frac{\partial^2}{\partial p_a \partial p_b}
 + \beta(a+b-1) p_a \,p_b \frac{\partial}{\partial p_{a+b-1}}
 \right) \\
 + \sum_{k=1}^{\infty}(\alpha+(1-\beta)(k+1)+2\beta N)\, k\, p_{k+1} \frac{\partial}{\partial p_k}
 + \beta N (\alpha + \beta(N-1)+1)\, p_1\,.
\end{multline}
\end{lemma}
\begin{proof}
First, we expand the degree-zero part of the recursion operator,
\begin{equation}
 1-x_0^{-} = -\hbar^2 \beta \hD + \frac12\hbar^3\beta\,\hW_0^\beta + O(\hbar^4)\,,
\end{equation}
with
\be
 \hW_0^\beta =
 \sum_{a,b=1}^\infty \left(a b p_{a+b} \frac{\partial^2}{\partial p_a \partial p_b}
 + \beta(a+b) p_a p_b \frac{\partial}{\partial p_{a+b}}\right)
 + (1-\beta) \sum_{k=1}^\infty k^2 p_k \frac{\partial}{\partial p_k}\,.
\ee
The operator $\hW_0^\beta$ is nothing but the Calogero Hamiltonian and corresponds to the $\psi_3$ generator of the affine Yangian in the Fock representation \cite{Matsuo:2023lky,Prochazka:2015deb}.
Next, we expand the order-one terms in $\hat{\sfA}_{1\bar{0}}$.
This can be done explicitly, but it is instructive to use the commutation relations \eqref{eq:[xpm,e1]}, i.e.\
\be
 x_{-1}^{\pm} = \left[\frac{1-x_0^{\pm}}{1-q^{\pm1}}, p_1 \right]
 = \pm\hbar\,\beta\, p_1 -\frac12\hbar^2 \beta \left(p_1-[\hW_0^\beta,p_1]\right)
 + O(\hbar^3)
\ee
where we used $[\hD,p_1]=p_1$.
We also recall that in the semiclassical limit we fix $u_1=-(1-q^{-1})a_1$.
Putting everything together, we find that the leading term in the $\hbar$-expansion of $\hat{\sfA}_{1\bar{0}}$ is given by \eqref{eq:classical-limit-D1}.
\end{proof}

We can also consider the simiclassical limit at the level of the gauge transformation operator in \eqref{eq:gauge-transf-Nf1}. We can compute its limit explicitly using the identity
\be
 \lim_{q\to1} \frac{\Delta^{+}_\lam(q^z)}{(1-q)^{|\lam|}}
 = \prod_{(i,j)\in\lam} (z+(j-1)-\beta(i-1))
\ee
hence we find 
\be
\ba
 \lim_{q\to1} \hG_{1\bar{0}}
 &= \lim_{q\to1} \frac{\Delta^{+}(Qt^{-1}qr)}{(1-q)^{\hD}}
 \frac{\Delta^{+}(Q)}{(1-q)^{\hD}}
 \exp\left(\sum_{k\geq1}\frac{(1-q)^{2k}u_1^{-k}p_k}{k(1-q^k)}\right)
 \frac{(1-q)^{\hD}}{\Delta^{+}(Q)} \frac{(1-q)^{\hD}}{\Delta^{+}(Qt^{-1}qr)}\\
 &= \hat{\sfC}^{\beta\mathrm{WL}}
 \exp\left(\frac{p_1}{a_1}\right)
 (\hat{\sfC}^{\beta\mathrm{WL}})^{-1}
 = \exp\left(\frac{\hW^\beta_{-1}}{a_1}\right)
\ea
\ee
where $\hat{\sfC}^{\beta\mathrm{WL}}$ is the operator diagonal on Jack functions with eigenvalues $\sfC_\lam^{\beta\mathrm{WL}} := \prod\limits_{(i,j)\in\lam} (N+(j-1)-\beta(i-1))(N-\beta+1+\alpha+(j-1)-\beta(i-1))$.
For $\beta=1$ and $a_1=1$, the operator $\hW^\beta_{-1}$ reduces to the operator $\hW_{-1}$ in \eqref{eq:gauge-transform-WL}.

We observe that the operator $\hW^{\beta}_{-1}$ can be written as a linear combination of the generators $e_0,e_1,e_2$ of the affine Yangian. This, along with the degeneration of Macdonald polynomials into Jack polynomials, supports the expectation, that the $\DIM$ symmetry of the matrix model reduces to the affine $\mathfrak{gl}_1$ Yangian in the semiclassical limit, as predicted on the algebraic level \cite{tsymbaliuk2017affine,Matsuo:2023lky}.

\section{Quantum toroidal \texorpdfstring{$\mathfrak{gl}_1$}{gl1} algebra and its horizontal Fock representation}
\label{sec:AppendixDIM}

Within the scope of this paper, we work only with a single representation of the quantum toroidal $\mathfrak{gl}_1$ algebra --- the Fock space of a single free boson field. Moreover, only a few, albeit special, elements of the algebra are going to play a role in our discussion. However, as we will see, some of the computations could be potentially extended to the whole algebra. So, for this reason and for completeness, we give here the general definition of the algebra, and its horizontal Fock representation. We will use the same notation for elements of the algebra and their representations on $\Lambda$ in order to simplify the exposition. This should not lead to confusion since we will only work in a single representation.
\begin{definition}
The algebra $\DIM$ is defined by the generating currents $x^{\pm}(z)$, $\psi^{\pm}(z)$:
\begin{equation}
    x^{\pm}(z) = \sum_{n \in \mathbb{Z}} x^{\pm}_n z^{-n}\,,
    \hspace{30pt}
    \psi^{\pm}(z) = \sum_{n \geq 0} \psi^{\pm}_{\pm n} z^{\mp n}\,,
\end{equation}
and two central elements $c,\bar{c}$, subject to the relations:
\begin{equation}
\label{eq:comm}
\begin{gathered}
 {\left[\psi^{\pm}(z), \psi^{\pm}(w)\right] = 0,
 \quad
 \psi^{+}(z) \psi^{-}(w)
 = \frac{g\left((t/q)^c \frac{w}{z}\right)}
 {g\left(\frac{w}{z}\right)}
 \psi^{-}(w) \psi^{+}(z),} \\
 \psi^{+}(z) x^{+}(w) = g\left(\frac{w}{z}\right)^{-1}
 x^{+}(w) \psi^{+}(z),
 \quad
 \psi^{-}(z) x^{+}(w)
 = g\left(\frac{z}{w}\right)
 x^{+}(w) \psi^{-}(z), \\
 \psi^{+}(z) x^{-}(w)
 = g\left((t/q)^c \frac{qw}{z}\right)
 x^{-}(w) \psi^{+}(z),
 \quad
 \psi^{-}(z) x^{-}(w)
 = g\left(\frac{z}{qw}\right)^{-1}
 x^{-}(w) \psi^{-}(z), \\
 {\left[x^{+}(z), x^{-}(w)\right] = \frac{(1-q)(1-t^{-1})}{(1-qt^{-1})}
 \left(
  \delta\Big((t/q)^{-c}\frac{z}{qw}\Big) \psi^{+}(z)
 -\delta\Big(\frac{z}{qw}\Big) \psi^{-}(z)
 \right),} \\
 G^{\mp}\left(\frac{z}{w}\right) x^{\pm}(z) x^{\pm}(w) =
 G^{\pm}\left(\frac{z}{w}\right) x^{\pm}(w) x^{\pm}(z),
\end{gathered}
\end{equation}
where $G^\pm(z):=(1-q^{\pm1}z)(1-t^{\mp1}z)(1-t^{\pm1}q^{\mp1}z)$ and $g(z):=G^+(z)/G^-(z)$.
Moreover, we have $\psi^\pm_0=(t/q)^{\mp\frac{\bar{c}}{2}}$.
We omit here the Serre relations for triple commutators of $x^\pm(z)$ for brevity, and we refer the reader to \cite{Ding:1996mq,Feigin_2009} for more details.%
\footnote{Notice that we have used slightly different definitions here compared to those references. Namely $x^-_{\text{here}}(z)=x^-_{\text{there}}((t/q)^{\frac{c}{2}}qz)$ and $\psi^\pm_{\text{here}}(z)=\psi^\pm_{\text{there}}((t/q)^{\mp\frac{c}{4}}z)$.}
\end{definition}

Fock representations of the algebra $\DIM$ are labeled by the values of the two central elements $(c,\bar{c})$, and are all isomorphic to $\Lambda$ as vector spaces.
We are interested in the horizontal Fock representation of levels $(c,\bar{c})=(1,0)$.
In this representation, the Cartan currents $\psi^\pm(z)$ act as certain exponentials in the Heisenberg algebra generators $p_k$ and $k\frac{1-q^k}{1-t^k}\frac{\partial}{\partial p_k}$ on the ring of symmetric functions $\Lambda=\BC[p_1,p_2,p_3,\dots]$,
\be\label{eq:VertexOperatorsPsi}
 \psi^+(z) = \exp\Big(-\sum_{k\geq1}(1-q^k)(1-t^kq^{-k})z^{-k}
 \frac{\partial}{\partial p_k}\Big)\,,
 \hspace{10pt}
 \psi^-(z) = \exp\Big(\sum_{k\geq1}(1-t^{-k})(1-t^kq^{-k})z^k\frac{p_k}{k}\Big)\,,
\ee
while the currents $x^\pm(z)$ act as the vertex operators\footnote{We remind, that we use the same symbol for the element of the algebra and its image in the representation $(c,\bar{c})=(1,0)$.} in \eqref{eq:VertexOperatorsx}, whose zero-modes are diagonal on the Macdonald basis.
The creation and annihilation operators in the Heisenberg algebra also act in a special way on Macdonald polynomials. In particular, the action of multiplication by $p_1$ is known as the Pieri rule \cite[Ch.VI, \textsection6]{Macdonald:book}: 
\be
    p_1 P_\lam = \sum_{\nu \in \mathrm{A}(\lam) } \psi_{\nu/\lam}(q,t) P_\nu\,,
\ee
where $\mathrm{A}(\lam)$ denotes the set of all partitions that can be obtained by adding one box to $\lam$ and the coefficient is given by
\be
 \psi_{\nu/\lam}(q,t) = \frac{p_1(\ve_\vac)P_\lam(\ve_\vac)}
 {P_\nu(\ve_\vac)} \res_{z=\chi_{\nu/\lam}}z^{-1}
 \exp\left(\sum_{k\geq1}\frac{z^{-k}}{k}p_k(x_\lam)\right)\,.
\ee
Using the explicit formulas \eqref{eq:VertexOperatorsx} for the vertex operators representation of $x^\pm(z)$ we can also derive the useful commutation relations
\be
\label{eq:[xpm,e1]}
 [x^\pm(z),p_1] = -(1-q^{\pm1})z^{-1}x^\pm(z)\,,
 \hspace{30pt}
 \left[x^\pm(z),\frac{\partial}{\partial p_1}\right]
 = -(1-t^{\mp1})zx^\pm(z)\,.
\ee

The algebra $\DIM$ is know to admits an $\mathrm{SL}(2,\BZ)$ group of outer automorphisms that permute the generators among each other. The $S$ generator of this group acts as the so called \emph{Miki automorphism} that, in the level-one Fock representation, sends $p_1^\perp$ to $x^+_0$, $x^+_0$ to $p_1$, $p_1$ to $x^-_0$ and $x^-_0$ to $p_1^\perp$ (up to multiplicative constants).
On the other hand, the $T$ generator acts as conjugation by the framing operator $\sfT$ defined in \eqref{eq:framing}. In fact, we have the following.
\begin{lemma}
The framing operator $\sfT$ satisfies the relations \cite[Proposition~1.5]{Garsia2001}
\label{lemma:T}
\be
\label{eq:idTe1T-}
 \sfT^{-1}x^+_{-1}\sfT = (1-t^{-1})\, p_1\,,
\ee
\be
\label{eq:idT-e1T}
 \sfT x^-_{-1}\sfT^{-1} = (1-t)\,  p_1\,,
\ee
\be
 \sfT x^+_{1}\sfT^{-1}
 = -(1-q)\, \dfrac{\partial}{\partial p_1}\,,
\ee
\be
 \sfT^{-1}x^-_{1}\sfT
 = -(1-q^{-1})\,\dfrac{\partial}{\partial p_1}\,.
\ee
\end{lemma}

Another important operator in the theory of Macdonald functions, is the delta operator defined in \eqref{eq:DeltaOp} which can be realized through linear combinations of horizontal generators of the algebra.
One of the main properties that it satisfies is the so called \emph{five-term relation} of \cite{Garsia:2018fiv} (see also \cite{Zenkevich:2021aus,Hu:2023dno,Hu2024:pro,Dali:2024mac,Bourgine:2025fwy}).
\begin{theorem}[Garsia--Mellit]
The delta operator $\Delta^{+}(z)$ satisfies the five-term relation
\be
\label{eq:5t}
 \exp\left(-\sum_{k\geq1} \frac{w^k p_k}{k(1-q^k)} \right)
 \Delta^{+}(z)
 \exp\left(\sum_{k\geq1} \frac{w^k p_k}{k(1-q^k)} \right)
 \Delta^{+}(z)^{-1}
 =
 \sfT\,
 \exp\left(\sum_{k\geq1} \frac{(-zw)^k p_k}{k(1-q^k)} \right)
 \sfT^{-1}\,.
\ee
\end{theorem}

\begin{remark}
The five-term relation for the Fock representation is conjecturally a consequence of a more general five-term relation in the algebra itself. In particular, $\Delta^{+}(z)$ can be given an explicit expression in terms of the generators of the algebra (see \cite{Zenkevich:2021aus}). This might hint that the techniques below, that heavily rely on this identity might be potentially applicable in a more general setup.
\end{remark}

\begin{corollary}\label{lemma:Delta}
From the five-term relation we obtain the identities
\be\label{eq:5t:e1}
 \Delta^{+}(z)\,p_1\,\Delta^{+}(z)^{-1} = p_1-z\dfrac{x_{-1}^{+}}{1-t^{-1}}\,,
\ee
\be\label{eq:5t:e2}
 \Delta^{+}(z)\,q^{-1}x^-_{-1}\,\Delta^{+}(z)^{-1} = q^{-1}x^-_{-1}+z(1-t^{-1})(t/q)p_1\,,
\ee
\be\label{eq:5t:e3}
 \Delta^{+}(z)\,q^{-2}x^-_{-2}\,\Delta^{+}(z)^{-1}
 = q^{-2}x^-_{-2}-\frac{z}{1-q/t} h_2\left((1-t^{-k})(1-(t/q)^{k})p_k\right)
 -z^2(t/q)x^+_{-2}\,,
\ee
\be\label{eq:5t:e4}
 \Delta^{+}(z)^{-1}\, p_1^\perp \,\Delta^{+}(z)
 = p_1^\perp
 + z\frac{x_1^{+}}{1-t}\,,
\ee
\be\label{eq:5t:e5}
 \Delta^{-}(z^{-1})\, p_1 \,\Delta^{-}(z^{-1})^{-1}
 = p_1
 - z^{-1}\frac{x_{-1}^{-}}{1-t}\,.
\ee
\end{corollary}
\begin{proof}
The relation \eqref{eq:5t:e1} is obtained by rewriting the five-term relation as
\be
 \Delta^{+}(z)
 h_n\left(\frac{p_k}{1-q}\right)
 \Delta^{+}(z)^{-1}
 = \sum_{i+j=n} (-z)^j
 h_i\left(\frac{p_k}{1-q^k}\right)
 \sfT\,
 h_j\left(\frac{p_k}{1-q^k}\right)
 \sfT^{-1}\,,
\ee
where we moved the first exponential in the l.h.s.\ (i.e.\ $T_{1,0}(w)^{-1}$) to the r.h.s.\ and we expanded to degree $n$ in powers of $w$ both sides.
By fixing $n=1$, and then using \eqref{eq:idT-e1T} to simplify the term $\sfT\, p_1 \sfT^{-1}$, we find \eqref{eq:5t:e1}.
Equation \eqref{eq:5t:e4} follows from \eqref{eq:5t:e1} by taking the adjoint of each term w.r.t.\ the Macdonald inner product.
Similarly, equation \eqref{eq:5t:e5} follows from \eqref{eq:5t:e1} by sending the parameters $q,t$ to their inverses $q^{-1},t^{-1}$.

The relation \eqref{eq:5t:e2} is obtained by observing that if we represent $x^-_{-1}$ as $\sfT^{-1}p_1\sfT$, then we can commute the action of $\sfT$ and $\Delta^{+}(z)$ (since they are both diagonal). Then we can use the previous relation to write
\be
\ba
 \Delta^{+}(z) x^-_{-1} \Delta^{+}(z)^{-1} &= (1-t)\Delta^{+}(z) \sfT^{-1} p_1\sfT \Delta^{+}(z)^{-1} \\
 &= (1-t) \sfT^{-1} \Delta^{+}(z) p_1 \Delta^{+}(z)^{-1}\sfT \\
 &= (1-t) \sfT^{-1} \Big(p_1-z\frac{x_{-1}^{+}}{1-t^{-1}}\Big) \sfT \\
 &= (1-t) \sfT^{-1} p_1 \sfT -z(1-t) \sfT^{-1} \frac{x_{-1}^{+}}{1-t^{-1}} \sfT \\
 &= x^-_{-1} -z(1-t) p_1\,.
\ea
\ee

Finally, we consider the l.h.s.\ of \eqref{eq:5t:e3}.
In order to compute $\Delta^{+}(z) x^-_{-2} \Delta^{+}(z)^{-1}$, we first represent $x^-_{-2}$ as the commutator $[p_1,x^-_{-1}]$ (as in \eqref{eq:[xpm,e1]}) and then we bring conjugation by $\Delta^{+}(z)$ inside the commutator:
\be
\ba
 &\Delta^{+}(z) q^{-2} x^-_{-2} \Delta^{+}(z)^{-1}
 = \frac{q^{-2}}{1-q^{-1}} \Delta^{+}(z) [p_1,x^-_{-1}] \Delta^{+}(z)^{-1}\\
 &= \frac{q^{-2}}{1-q^{-1}} \left[\Delta^{+}(z) p_1 \Delta^{+}(z)^{-1} , \Delta^{+}(z)x^-_{-1} \Delta^{+}(z)^{-1} \right] \\
 &= \frac{q^{-2}}{1-q^{-1}} \left[ p_1-z\frac{x_{-1}^{+}}{1-t^{-1}} , x^-_{-1} -z(1-t) p_1 \right] \\
 &= \frac{q^{-2}}{1-q^{-1}}\left( [p_1,x^-_{-1}] -z(1-t) [p_1,p_1]
 -\frac{z}{1-t^{-1}} [ x_{-1}^{+} , x^-_{-1} ]
 -t z^2 [ x_{-1}^{+} , p_1 ] \right) \\
 &=  q^{-2} x^-_{-2}
 -\frac{z}{1-q/t} h_2\left((1-t^{-k})(1-(t/q)^{k})p_k\right)
 - z^2 (t/q) x_{-2}^{+} \\
\ea
\ee
where we used the commutation relations \eqref{eq:[xpm,e1]} and \eqref{eq:comm}.
\end{proof}

\section{Relation to multivariable Al-Salam--Carlitz polynomials}
\label{app:AlSalamCarlitz}

The Al-Salam--Carlitz polynomials $U^{(a)}_\mu(\bx)$ correspond to a $q,t$-deformation of the classical Hermite
polynomials in the sense that they form an orthogonal basis w.r.t.\ the $q,t$-Gaussian weight function
\be
 w_U(x) = \frac{(qx;q)_\infty(qx/a;q)_\infty}
 {(q;q)_\infty(a;q)_\infty(q/a;q)_\infty}\,.
\ee
Clearly this is related to the weight $w(x)$ in \eqref{eq:GeneralNfmodel} for $\Nf=2$, $\Nbf=0$, where one specializes $u_1=1$, $u_2=1/a$, $\ell=0$ and $r=1$. As a result, it follows that the inhomogeneous functions $\sfW^{2\bar0}_\mu(\bx)$ defined in \eqref{eq:def-WZ} should match with the Al-Salam--Carlitz polynomials after restriction to a finite number of variables and specializations of the parameters $u_{1,2}$. In this section, we will show that indeed the two coincide, by showing that they both provide a basis of eigenfunctions for the same $q$-difference operator.

In \cite{Baker:1997mul}, it is shown that the polynomials $U^{(a)}_\mu(\bx)$ are eigenfunctions of the operator
\be
\label{eq:ASC-op}
 \mathcal{H} = u^{-}_0-(1+a)[E_0,u^{-}_0]+a[E_0,[E_0,u^{-}_0]]\,,
\ee
where
\be
 E_0
 = \sum_{i=1}^N \prod_{j\neq i}\frac{tx_i-x_j}{x_i-x_j} D^q_{x_i}
 = \frac1{1-q}\left(p_1(\bx^{-1}) - u^+_{1}\right)
\ee
and we introduced a generalization of the operators $u^\pm_0$,
\be
 u^\pm_{-k} := \sum_{i=1}^N x_i^k\prod_{j\neq i}\frac{t^{\pm1}x_i-x_j}{x_i-x_j}
 q^{\pm x_i\frac{\partial}{\partial x_i}}\,,
\ee
which are homogeneous of degree $k$ in the $\bx$-variables.
The eigenvalue equation then reads
\be
 \mathcal{H}\cdot U^{(a)}_\mu(\bx) = u^\vee_\mu\, U^{(a)}_\mu(\bx)
\ee
with eigenvalue $u_\mu^\vee$ as in \eqref{eq:symbols3}.

In order to match with the eigenvalue equation in \eqref{eq:eigenop-WZ-}, we first need to ``\emph{bosonize}'' the operator $\mathcal{H}$, i.e.\ we have to take the stable limit $N\to\infty$. In this limit, one has
\be
 \lim_{N\to\infty} u^\pm_{-k}
 = \frac{\delta_{k,0}-Q^{\pm1}x^\pm_{-k}}{1-t^{\pm1}}
\ee 
and similarly, $\lim_{N\to\infty}p_1(\bx^{-1})=p_1^\perp$.
We can now compute the commutators in the r.h.s.\ of \eqref{eq:ASC-op} directly using the $\DIM$ algebra relations. The relevant commutation relations are
\be
\ba
 {}[p_1^\perp, x^-_0] &= (1-q) x^-_1\,, \\
 [x_1^+, x^-_0] &= -q^{-1}(1-q)(1-t)^2 p_1^\perp\,, \\
 [p_1^\perp, x^-_1] &= (1-q) x^-_2\,, \\
 [x_1^+, x^-_1] &= -\frac{(1-q)(1-t)}{(1-q/t)}
 h_2\left((1-t^{-k})(1-(t/q)^k)p_k^\perp\right)\,, \\
 [x^+_1, p_1^\perp] &= (1-q) t^{-1} x^+_2\,.
\ea
\ee
We can then write
\begin{multline}
\label{eq:limH-umu1}
 \lim_{N\to\infty}(\mathcal{H}-u^\vee_\mu)
 = \frac{Q^{-1}}{1-t^{-1}} \Big\{
 (x^\vee_\mu-x^{-}_0)
 + (1+a)
 \left(x^-_1+Q(1-t^{-1})(t/q)p_1^\perp\right)\\
 -a\left(x^{-}_2-\frac{Q}{1-q/t}h_2\left((1-t^{-k})(1-(t/q)^k)p_k^\perp\right)
 -Q^2(t/q)t^{-2}x^{+}_2\right)
 \Big\}
\end{multline}
and
\begin{multline}
\label{eq:limH-umu2}
 \lim_{N\to\infty}(\mathcal{H}-u^\vee_\mu)^\perp
 = \frac{Q^{-1}}{1-t^{-1}} \Big\{
 (x^\vee_\mu-x^{-}_0) + (1+a)\left(q^{-1}x^{-}_{-1}+Q(1-t^{-1})(t/q)p_1\right)\\
 -a\left(q^{-2}x^{-}_{-2}-\frac{Q}{1-q/t}h_2\left((1-t^{-k})(1-(t/q)^k)p_k\right)
 -Q^2(t/q)x^{+}_{-2}\right)
 \Big\}
\end{multline}
where we used that $(x^+_{k})^\perp=t^{k}x^+_{-k}$ and $(x^-_{k})^\perp=q^{-k}x^-_{-k}$.
It is now straightforward to check that both \eqref{eq:limH-umu1} and \eqref{eq:limH-umu2} have one-dimensional kernels corresponding to the functions $\sfW^{2\bar0}_\mu$ and $\sfZ^{2\bar0}_\mu$, respectively, after specializing to $u_1=1$ and $u_2=1/a$.
This shows that the stable limit of the polynomials $U^{(a)}_\mu(\bx)$ coincides with the symmetric functions $\sfW^{2\bar0}_\mu(\bp)$, since they both provide eigenbasis for the same operator, and they satisfy the same normalization condition
\be
 \langle P_\mu,U^{(a)}_\mu\rangle_{q,t} = b_\mu^{-1} = \langle P_\mu,\sfW^{2\bar0}_\mu\rangle_{q,t}\,.
\ee
Moreover, the eigenoperator \eqref{eq:limH-umu2} for $\mu=\vac$, coincides with the recursion operator $\hat{\sfA}^{2\bar0}$ in \eqref{eq:D2}.

We can even directly compare the $q$-exponential operator formula of \cite{Baker:1997mul} with the gauge transformation operator $\hG_{2\bar{0}}$ in \eqref{eq:G20}. Following the notation used in \cite{Baker:1997mul}, we have the function
\be
 \rho_a(x;q) = (qx;q)_\infty(qx/a;q)_\infty
\ee
which we can use to write
\be
 \lim_{N\to\infty} \prod_{i=1}^N \frac{1}{\rho_a(x_i;q)}
 = \exp\left(\sum_{k\geq1}\frac{(1+a^k)p_k}{k(1-q^k)}\right)
 = \Delta^{+}(Q)^{-1}\cdot\hG_{2\bar{0}}\cdot\Delta^{+}(Q)
\ee
According to \cite[(3.9)]{Baker:1997mul}, the $q$-exponential operator, i.e.\ the finite-variable version of the operator $(\hG_{2\bar{0}}^\perp)^{-1}$, is given by
\be
 \prod_{i=1}^N \rho_a(D_i;q)
 = \exp\left(-\sum_{k\geq1}\frac{(1+a^k)}{k(1-q^k)}\sum_{i=1}^N D_i^k\right)
 \xrightarrow{N\to\infty} (\hG_{2\bar{0}}^\perp)^{-1}
\ee
where $D_i$ are the (A-type) $q$-Dunkl operators. Comparing these two formulas, we obtain
\be
 \lim_{N\to\infty} f(D_1,\dots,D_N)
 = \Delta^{+}(Q)^{-1}\cdot f^\perp \cdot\Delta^{+}(Q)
\ee
for any symmetric function $f\in\Lambda$.\footnote{For $f=p_1$, this is the operator in \eqref{eq:5t:e4}.} Defining the hypergeometric function
\be
 {}_0\mathcal{F}_0(\bx,\by) := \sum_\lam b_\lam \frac{P_\lam(p_k=\frac1{1-t^k})}
 {P_\lam(p_k=\frac{1-Q^k}{1-t^k})} P_\lam(\bx) P_\lam(\by)
 = \left.\Delta^{+}(Q)^{-1}\right|_{(\bx)}\cdot\exp\left(\sum_{k\geq1}\frac{1-t^k}{1-q^k}\frac{p_k(\bx)p_k(\by)}{k}\right)
\ee
where $\left.\Delta^{+}(Q)\right|_{(\bx)}$ is the delta operator \eqref{eq:DeltaOp} acting on the $\bx$-variables only,
we then obtain
\be
\ba
 \lim_{N\to\infty}\left.f(D_1,\dots,D_N)\right|_{(\bx)}\cdot {}_0\mathcal{F}_0(\bx,\by)
 &= \left.\Delta^{+}(Q)^{-1}\cdot f^\perp\right|_{(\bx)} \cdot
 \exp\left(\sum_{k\geq1}\frac{1-t^k}{1-q^k}\frac{p_k(\bx)p_k(\by)}{k}\right) \\
 &= \left.\Delta^{+}(Q)^{-1}\right|_{(\bx)}\cdot f(\by) \cdot
 \exp\left(\sum_{k\geq1}\frac{1-t^k}{1-q^k}\frac{p_k(\bx)p_k(\by)}{k}\right) \\
 &= f(\by) \, {}_0\mathcal{F}_0(\bx,\by)
\ea
\ee
which is the stable limit of \cite[(3.8)]{Baker:1997mul}.

It is now evident that for any other model we can define the associated hypergeometric function
\be
 \mathcal{F}^w(\bx,\by) :=
 \left.(\hat{\sfC}^w)^{-1}\right|_{(\bx)}\cdot\exp\left(\sum_{k\geq1}\frac{1-t^k}{1-q^k}\frac{p_k(\bx)p_k(\by)}{k}\right)
\ee
and then superintegrability can be equivalently restated as the identity
\be
 \ev{ \mathcal{F}^w(\bx,\by) }^w = \exp\left(\sum_{k\geq1}\frac{1-t^k}{1-q^k}\frac{\varphi^w_k\, p_k(\by)}{k}\right)\,.
\ee

\bibliographystyle{utphys}
\bibliography{references}

\end{document}